\documentclass[a4paper,11pt]{article}
\pdfoutput=1
\RequirePackage{ifluatex}
\usepackage{fullpage}
\usepackage[utf8]{inputenc}
\usepackage[T1]{fontenc}
\usepackage{graphicx}
\usepackage{amsmath,bm}
\usepackage{float}
\usepackage{amsfonts}
\usepackage{algorithm}      
\usepackage[noend]{algpseudocode}  
\usepackage{algorithmicx}
\usepackage{authblk}
\usepackage[hypertexnames=false,colorlinks=true,urlcolor=Blue,citecolor=Green,linkcolor=BrickRed]{hyperref}
\usepackage[dvipsnames,usenames]{xcolor}
\usepackage{amsthm}
\usepackage[noadjust]{cite}
\usepackage{todonotes}
\usepackage{thmtools,thm-restate}
\usepackage{lineno}
\usepackage{comment}

\newcommand{\Oh}{\mathcal{O}}
\newcommand{\tildeOh}{\mathcal{\tilde O}}
\newcommand{\eps}{\epsilon}
\newcommand{\lam}{\lambda}

\newcommand{\CF}{\mathsf{CF}}

\usepackage{tabularx} 
\newcolumntype{C}{>{\centering\arraybackslash}X} 

\newtheorem{theorem}{Theorem}
\newtheorem{lemma}[theorem]{Lemma}

\newtheorem{definition}[theorem]{Definition}

\title{Fully Dynamic Approximation of LIS in Polylogarithmic Time}

\author[ ]{Paweł Gawrychowski\thanks{\texttt{\href{mailto:gawry@cs.uni.wroc.pl}{gawry@cs.uni.wroc.pl}}}}
\author[ ]{Wojciech Janczewski\thanks{\texttt{\href{mailto:wojciech.janczewski@cs.uni.wroc.pl}{wojciech.janczewski@cs.uni.wroc.pl}}}}
\affil[ ]{University of Wrocław}

\date{}

\begin{document}

\maketitle

\begin{abstract}
We revisit the problem of maintaining the longest increasing subsequence (LIS) of an array
under (i) inserting an element, and (ii) deleting an element of an array. In a recent breakthrough,
Mitzenmacher and Seddighin [STOC 2020] designed an algorithm that maintains an
$\Oh((1/\epsilon)^{\Oh(1/\epsilon)})$-approximation of LIS under both operations with worst-case
update time $\tildeOh(n^{\epsilon})$\footnote{$\tildeOh$ hides factors polynomial in $\log n$,
where $n$ is the length of the input.}, for any constant $\epsilon>0$. We exponentially improve
on their result by designing an algorithm that maintains an $(1+\epsilon)$-approximation of LIS
under both operations with worst-case update time $\tildeOh(\eps^{-5})$. 
Instead of working with the grid packing technique introduced by Mitzenmacher and Seddighin,
we take a different approach building on a new tool that might be of independent interest: LIS sparsification.

A particularly interesting consequence of our result is an improved solution for the so-called
Erdős-Szekeres partitioning, in which we seek a partition of a given permutation of $\{1,2,\ldots,n\}$
into $\Oh(\sqrt{n})$ monotone subsequences. This problem has been repeatedly stated as
one of the natural examples in which we see a large gap between the decision-tree complexity
and algorithmic complexity. The result of Mitzenmacher and Seddighin implies an $\Oh(n^{1+\epsilon})$
time solution for this problem, for any $\epsilon>0$. Our algorithm (in fact, its simpler decremental
version) further improves this to $\tildeOh(n)$.
\end{abstract}

\thispagestyle{empty}
\clearpage
\setcounter{page}{1}

\section{Introduction}

Computing the length of a longest increasing subsequence (LIS) is one of the basic algorithmic
problems. Given a sequence $(a_{1},a_{2},\ldots,a_{n})$ with a linear order on the elements, an increasing subsequence
is a sequence of indices $1\leq i_{1} < i_{2} < \ldots < i_{\ell}\leq n$ such that $a_{i_{1}} < a_{i_{2}} < \ldots < a_{i_{\ell}}$.
We seek the largest $\ell$ for which such a sequence exists. It is well known that $\ell$ can be computed
in $\Oh(n\log n)$ time using either dynamic programming and an appropriate data structure or
by computing the first row of the Young tableaux (see e.g.~\cite{Fredman}, we will refer to this procedure
as Fredman's algorithm even though Fredman himself attributes it to Knuth). The latter method
admits an elegant formulation as a card game called patience sorting, see~\cite{AD99}.
For comparison-based algorithms, a lower bound of $\Theta(n\log n)$ was shown by Fredman~\cite{Fredman},
and later modified to work in the more powerful algebraic decision tree model by Ramanan~\cite{Ramanan97}.
However, under the natural assumption that the elements in the array belong to $[n]$, an
$\Oh(n\log\log n)$ time Word RAM algorithm can be obtained using a faster data structure such
as van Emde Boas trees~\cite{Boas77}. This solution has been further refined to work in $\Oh(n\log\log k)$ time,
where $k$ is the answer, by Crochemore and Porat~\cite{CrochemoreP10}.

\paragraph{Dynamic algorithms. }

Even linear (or almost-linear) algorithms are too slow when dealing with large data. This is particularly
relevant when we need to repeatedly query such data that undergoes continuous updates. One possible
approach is then to design a dynamic algorithm capable of maintaining the answer during such modifications.
The updates could be simply appending new items, or inserting/substituting/deleting an arbitrary item.

There exists a long line of work on dynamic algorithms for graph problems, in which the updates
consist in adding/removing edges. Examples of problems considered in this setting include 
fully dynamic maximal independent set~\cite{MIS1,MIS2,MIS3}, fully dynamic minimum spanning
forest~\cite{Forest,Wulff-Nilsen17,NanongkaiS17,HolmLT01}, fully dynamic matching~\cite{BhattacharyaCH20,BernsteinHR19,BernsteinS16},
fully dynamic APSP~\cite{GutenbergW20b}, incremental/decremental reachability and single-source
shortest paths~\cite{BernsteinPW19,GutenbergWW20,GutenbergW20a,GutenbergW20},
or fully dynamic Steiner tree~\cite{LackiOPSZ15}. We stress that many of these papers resort to maintaining
an approximate solution, and that (in most cases) the goal is to achieve polylogarithmic update/query
time.

While by now we have nontrivial and surprising dynamic algorithms for many problems, in some cases
even allowing randomisation and amortisation does not seem to help. Abboud and Dahlgaard~\cite{AbboudD16}
showed how to use the popular conjectures to provide an evidence that, for some problems,
subpolynomial update/query solutions are unlikely, even for very restricted classes of graphs.
\cite{OnlineMV} introduced a new conjecture tailored for showing that, for multiple
natural dynamic problems, subpolynomial update/query time would be surprising (it has been later
shown that the online Boolean matrix-vector problem considered in this conjecture does admit a very
efficient cell probe algorithm~\cite{LarsenW17,ChakrabortyKL18}, so one cannot hope to prove it by
purely information-theoretical methods). Therefore, it seems that by now we have both a number
of nontrivial algorithms and some tools for proving that, in some cases, such an algorithm would be
very surprising.

Fredman's algorithm can be used to maintain the length of LIS under appending elements to the sequence
in $\Oh(\log n)$ time (or, by symmetry, prepending). However, it is already unclear if we can support
both appending and prepending elements in the same time complexity. Chen, Chu and Pinsker~\cite{ExactDynamic}
considered the more general question of maintaining the length of LIS under insertions and deletions of
elements in a sequence of length $n$, and showed how to implement both operations in $\Oh(\ell\log(n/\ell))$ time,
where $\ell$ is the current length of LIS. In the worst case, this could be linear in $n$, so only slightly better
than recomputing from scratch.  
In a recent breakthrough, Mitzenmacher and Seddighin~\cite{Grids} overcame this obstacle by relaxing
the problem and maintaining an approximation of LIS. For any constant $\epsilon>0$, their algorithm
maintains an $\Oh((1/\epsilon)^{\Oh(1/\epsilon)})$-approximation of LIS in $\tildeOh(n^{\epsilon})$ time
per an insertion or deletion.
Their solution is of course a nontrivial improvement on the worst-case $\Theta(n)$ time complexity, but comes
at the expense of returning an approximate solution.
This also brings the challenge of determining if we can improve the update time to, say, polylogarithmic
(which seems to be the natural complexity for a dynamic algorithm), and determining the dependency
on $\epsilon$.
Very recently, Kociumaka and Seddighin~\cite{Kociumaka20} presented the first exact fully dynamic
LIS algorithm with sublinear update time $\tildeOh(n^{4/5})$.

\paragraph{Different models of computation.}

In the streaming model, it is usual to mostly focus on the working space of an algorithm instead of its
running time.
The distance to monotonicity (DTM) of a sequence is the minimum number of edit operations required
to make it sorted. This is easily seen to be the length of the sequence
minus the length of its LIS. In the streaming model, computing both LIS and DTM requires $\Omega(n)$
bits of space, so it is natural to resort to approximation algorithms. For DTM, \cite{Stream2} gave a
randomised $(4+\epsilon)$-approximation in $\Oh(\log^{2}n)$ bits of space, and \cite{SaksS13} improved
this to $(1+\epsilon)$-approximation in $\Oh(1/\epsilon\cdot \log^{2}n)$ bits of space.
\cite{Stream2} also provided a deterministic $(1+\epsilon)$-approximation in $\Oh(\sqrt{n})$ space,
and \cite{ErgunJ08} gave a deterministic $(2+o(1))$-approximation in polylogarithmic space. Finally,
\cite{NaumovitzS15} provided a deterministic $(1+\epsilon)$-approximation in polylogarithmic space.
For LIS, \cite{Stream2} provided a deterministic $(1+\epsilon)$-approximation in $\Oh(\sqrt{n})$ space,
and this was later proved to be essentially the best possible~\cite{ErgunJ08,Stream1}.

In the property testing model, \cite{Approx} showed how to approximate LIS to within an additive
error of $\epsilon\cdot n$, for an arbitrary $\epsilon\in (0,1)$, with $\tildeOh((1/\epsilon)^{1/\epsilon})$ queries.
Denoting the length of LIS by $\ell$,
\cite{RubinsteinSSS19} designed a nonadaptive $\Oh(\lambda^{-3})$-approximation algorithm, where
$\lambda=\ell/n$, with $\tildeOh(\lambda^{-7}\sqrt{n})$ queries (and also obtained different
tradeoffs between the dependency on $\lambda$ and $n$). Very recently, \cite{Ilan} proved that
adaptivity is essential in obtaining polylogarithmic query complexity (with the exponent independent of
$\epsilon$) for this problem.

In the read-only random access model, \cite{Space} showed how to find the length of LIS in $\Oh(n^{2}/s\cdot\log n)$
time and only $\Oh(s)$ space, for any parameter $\sqrt{n}\leq s \leq n$. Investigating the time complexity
for smaller values of $s$ remains an intriguing open problem.

\paragraph{Related work. }

LIS can be seen as a special case of the longest common subsequence. LCS is another fundamental
algorithmic problem, and it has received significant attention. A textbook dynamic programming
solution allows calculating LCS of two sequences of length $n$ in $\Oh(n^{2})$ time, and the so-called
``Four Russians'' technique brings this down to $\Oh(n^{2}/\log^{2}n)$ for constant alphabets~\cite{MasekP80}
and $\Oh(n^{2}(\log\log n)^{2}/\log^{2}n)$~\cite{BilleF08} or even $\Oh(n^{2}\log\log n/\log^{2}n)$~\cite{Grabowski16}  for general alphabets.
Recently, there was some progress in providing explanation for why a strongly subquadratic $\Oh(n^{2-\epsilon})$ time algorithm is unlikely~\cite{AbboudBW15,BringmannK15},
and in fact even achieving $\Oh(n^{2}/\log^{7+\epsilon}n)$ would have some exciting unexpected consequences~\cite{AbboudB18}.
This in particular implies that one cannot hope for a strongly sublinear time dynamic algorithm, even if only
appending letters to one of the strings is allowed (unless the Strongly Exponential Time Hypothesis is false).
Very recently, Charalampopoulos, Kociumaka and Mozes~\cite{Charalampopoulos20a} matched this conditional lower bound,
providing a fully dynamic algorithm with $\tildeOh(n)$ update time.

Given that LCS seems hard to solve in strongly subquadratic time, it is tempting to seek an approximate solution.
This turns out to be surprisingly difficult (in contrast to the related question of computing the edit distance,
for which by now we have approximation algorithms with very good worst-case guarantees, see~\cite{Andoni} and
the references therein). Only very recently Rubinstein and Song~\cite{RubinsteinS20} showed how to improve on the
simple 2-approximation for binary strings in strongly subquadratic time. 
Previously,~\cite{LCS} showed how to obtain $\Oh(n^{0.498})$-approximation in linear time
(this should be compared with the straightforward $\Oh(\sqrt{n})$-approximation).

\paragraph{Our results. }
We consider maintaining an approximation of LIS under insertions and deletions of elements in a sequence.
Denoting the current sequence by $(a_{1},a_{2},\ldots,a_{n})$, an insertion of an element $x$ at position $i$
transforms the sequence into $(a_{1},a_{2},\ldots,a_{i-1},x,a_{i},\ldots,a_{n})$, while a deletion of an element
at position $i$ transforms it into $(a_{1},a_{2},\ldots,a_{i-1},a_{i+1},\ldots,a_{n})$. Denoting by $\ell$
the length of LIS of the current sequence, we seek an algorithm that returns its $(1+\epsilon)$-approximation,
that is, a number from $[\ell/(1+\epsilon),\ell]$. Our main result is as follows.

\begin{theorem} \label{th:main}
For any $\epsilon>0$, there is a fully dynamic algorithm maintaining an $(1+\eps)$-approximation of LIS
with insertions and deletions working in $\Oh(\eps^{-5} \log^{11}n)$ worst-case time.
\end{theorem}

In fact, our algorithm allows for slightly more general queries, namely approximating LIS of any continuous
subsequence $(a_{i},a_{i+1},\ldots,a_{j})$, in $\Oh(\log^2 n)$ time. Furthermore, if the returned
approximation is $k$, then in time $\Oh(k)$ the algorithm can also provide an increasing subsequence of
length $k$. Finally, the algorithm can be initialised with a sequence of length $n$ in time $\Oh(n\eps^{-2}\log^{6}n)$. 

The time complexities of our algorithm should be compared with that of Mitzenmacher and Seddighin~\cite{Grids},
who provide an $\Oh((1/\epsilon)^{\Oh(1/\epsilon)})$-approximation in $\tildeOh(n^{\epsilon})$ worst-case time per
an insertion or deletion.
In this context, we provide an exponential improvement: instead of providing a constant approximation in
$\Oh(n^{\epsilon})$ time per update, for an arbitrarily small but constant $\epsilon>0$, we are able to provide $(1+\eps)$-approximation in polylogarithmic time.
This is obtained by introducing a new tool, called LIS sparsification, and taking a different approach than
the one based on grid packing described by Mitzenmacher and Seddighin.
We remark that, while they start with a simple $(1+\epsilon)$-approximation in $\tildeOh_{\epsilon}(\sqrt{n})$
time per update, the approximation factor of their main algorithm is never better than $2$ (in fact, it is much higher),
and this seems inherent to their approach.

\paragraph{Erdős-Szekeres partitioning. }

The well-known theorem of Erdős and Szekeres states that any sequence consisting of distinct elements with length at least $(r-1)(s-1)+1$
contains a monotonically increasing subsequence of length $r$ or a monotonically decreasing subsequence
of length $s$~\cite{Erdos}. In particular, any sequence consisting of $n$ distinct elements contains a monotonically increasing
or decreasing subsequence of length $\sqrt{n}$, and as a consequence any such permutation can be
partitioned into $\Oh(\sqrt{n})$ monotone subsequences (it is easy to see that this bound is asymptotically
tight). The algorithmic problem of partitioning such a sequence into $\Oh(\sqrt{n})$ monotone
subsequences is known as the Erdős-Szekeres partitioning. A straightforward application of Fredman's algorithm
gives an $\Oh(n^{1.5}\log n)$ time solution for this problem, and this has been improved to $\Oh(n^{1.5})$ by
Bar-Yehuda and Fogel~\cite{Bar-YehudaF98}. However, in the (nonuniform) decision tree model
we have a trivial solution that uses $\Oh(n\log n)$ comparisons: we only need to identify the sorting permutation
and then no further comparisons are required. In their breakthrough paper on the decision tree complexity
of \textsf{3SUM}, Grønlund and Pettie~\cite{GronlundP18} mention Erdős-Szekeres partitioning as one of the natural
examples of a problem with a large gap between the (nonuniform) decision tree complexity and the (uniform)
algorithmic complexity. Another examples include \textsf{3SUM}, for which the decision tree complexity was first
decreased to $\Oh(n^{1.5}\sqrt{\log n})$~\cite{GronlundP18} and then further to $\Oh(n\log^{2}n)$~\cite{KaneLM19},
and APSP, for which the decision tree complexity is known to be $\Oh(n^{2.5})$~\cite{Fredman76}.

Closing the gap between the decision and the algorithmic complexity of Erdős-Szekeres partitioning 
is interesting not only as an intriguing puzzle, but also due to its potential applications. Namely, we hope to use
it as a preprocessing step and achieve a speedup by operating on the obtained monotone subsequences.
Very recently, Grandoni, Italiano, Łukasiewicz, Parotsidis and Uznański~\cite{Aleksander} successfully applied such a strategy to design a faster solution
for the all-pairs LCA problem. The crux of their approach is a preprocessing step that partitions a given poset
into $\Oh(\ell)$ chains and $\Oh(n/\ell)$ antichains, for a given parameter $\ell$, in $\Oh(n^{2})$ time.
Of course, this can be directly applied to partition a sequence into $\Oh(\sqrt{n})$ monotone subsequences by setting
$\ell=\sqrt{n}$, but this does not constitute an improvement in the time complexity, and it is not clear whether
similar techniques can help here.
However, we can apply the recent result of Mitzenmacher and Seddighin~\cite{Grids} (in fact, only deletions
are necessary, but this does not seem to significantly simplify the algorithm) to obtain such a partition in $\Oh(n^{1+\epsilon})$ time,
for any constant $\epsilon>0$, as follows. We maintain a constant approximation of LIS in $\Oh(n^{\epsilon})$ time
per deletion. As long as the approximated length of LIS is at least $\sqrt{n}$, we extract the corresponding increasing
subsequence (it is straightforward to verify that the structure of Mitzenmacher and Seddighin does provide such an
operation in time proportional to the length of the subsequence) and delete all of its elements.
This takes $\Oh(n^{1+\epsilon})$ time overall and creates no more than $\sqrt{n}$ increasing subsequences.
When the approximated length drops below $\sqrt{n}$, we know that the exact length is $\Oh(\sqrt{n})$.
A byproduct of Fredman's algorithm for computing the length $k$ of LIS is a partition of the elements into
$k$ decreasing subsequences. Thus, by spending additional $\Oh(n\log n)$ time we obtain the desired partition
into $\Oh(\sqrt{n})$ monotone subsequences in $\Oh(n^{1+\epsilon})$ total time.
Plugging in our algorithm for maintaining $(1+\epsilon)$-approximation of LIS in $\tildeOh(\eps^{-5})$ time
per deletion with, say, $\epsilon=1$, we significantly improve this time complexity to $\tildeOh(n)$.
We note that our algorithm works in the comparison-based model, and in such model $\Omega(n\log n)$
comparisons are required (this essentially follows from Fredman's lower bound for LIS, but we provide the details
for completeness).

\paragraph{Parallel and independent work. }
Shortly after a preliminary version of our paper appeared on arXiv, two other relevant papers were made public.
First, Seddighin and Mitzenmacher~\cite{Mitzenmacher20} independently observed that their approximation algorithm
can be applied to obtain Erdős-Szekeres partition (in particular, they provide a detailed description of how to modify
their solution to extract the elements of LIS).
Second, Kociumaka and Seddighin~\cite{Kociumaka20} modified the grid packing technique to obtain
an algorithm with update time $\Oh(n^{o(1)})$ and approximation factor $1+o(1)$.
While their modification allows for more general queries, it is not able to provide constant approximation
in polylogarithmic time.

\paragraph{Overview of our approach. } 

We start with a description of the main ingredient of our improved solution in a static setting, in which
we are given an array $(a_{1},a_{2},\ldots,a_{n})$ and want to preprocess it for computing LIS in any subarray
$(a_{i},\ldots,a_{j})$, or $(i,j)$ for short. While it is known how to build a structure of size $\Oh(n\log n)$ capable of providing
exact answers to such queries in $\Oh(\log n)$ time using the so-called unit-Monge matrices~\cite[Chapter 8]{Tiskin},
this solution seems inherently static, and we follow a different approach that provides approximate
answers.

We require that the structure is able to return $(1+\epsilon)^2$-approximate solution for any subarray $(i,j)$.
To this end, it consists of $\log_{1+\epsilon}n$ levels. The purpose of level $k$ is to return, given
a subarray $(i,j)$ with LIS of length at least $(1+\epsilon)^{k+1}$, a subarray $(i',j')$ with
$i \leq i' \leq j' \leq j$ and LIS of length at least $(1+\epsilon)^{k}$.
This indeed allows us to approximate the length of LIS in a subarray $(i,j)$ by binary searching over the levels
to find the largest level $k$ for which the structure does not fail.
The information stored on each level could of course be just a sorted list of all minimal subarrays
$(i',j')$ with LIS of length $\lceil (1+\epsilon)^{k} \rceil$. 
This is, however, not very useful when we try to make the structure dynamic for the following
reason. Whenever we, say, delete an element at position $x$, this might possibly affect every level.
Now, for a level $k$ of the structure, there could be even $\Omega(n)$ minimal subarrays
containing $x$, or (what is even worse) using the element $a_{x}$ for their LIS. Such a situation
might repeat again and again, which makes obtaining even an amortised efficient solution problematic.
This suggests that we should maintain a sorted list of subarrays with small \emph{depth},
defined as the largest number of stored subarrays possibly containing the same position $x$.
Somewhat surprisingly, it turns out that there always exists such a sorted list of depth $\epsilon^{-1}$.
Our proof is constructive and based on a simple greedy procedure that actually constructs the list efficiently.
The sorted list of subarrays stored at each level is called a cover, while the whole structure
is referred to as a covering family.
We call this method of creating an approximation covering family of small depth the sparsification procedure.

To explain how this insight can be applied for dynamic LIS, first we focus on the decremental version
of the problem, in which we only need to support deletions.
This is already quite hard if we aim for polylogarithmic update time, and has interesting consequences.

We start with normalising the entries in the input sequence $(a_{1},a_{2},\ldots,a_{n})$ to form a permutation
of $[n]$ (if there are ties, earlier elements are larger as to preserve the length of LIS).
Then, it is helpful to visualise the input array $(a_{1},a_{2},\ldots,a_{n})$ as a set of points $S=\{(i,a_{i}) : i\in [n] \}$.
A deletion simply removes a point from $S$, and do not re-normalise the coordinates of the remaining points.
It is straightforward to translate a deletion of element at position $i$ of the current sequence into a deletion
of a specified point of $S$ in $\Oh(\log n)$ time.
Now, finding LIS in the current sequence translates into finding the longest \emph{chain} of points in
the current $S$, defined as an ordered subset of points with the next point strictly dominating the
previous point. In fact, our structure will implement more general queries corresponding to
finding LIS in any subarray $(a_{i},\ldots,a_{j})$ of the current array, or using a geometric interpretation
the longest chain in the subset of $S$ consisting of all points with the $x$-coordinate in a given interval $(x_{1},x_{2})$.

We maintain a recursive decomposition of $[n]\times [n]$ into smaller rectangles, roughly speaking
by applying a primary divide-and-conquer guided by the $y$-coordinates, and then a secondary
divide-and-conquer guided by the $x$-coordinates. Formally, we define dyadic intervals of the form $(i2^{k},(i+1)2^{k}-1)$,
and consider all rectangles of the form $R=(x_{1},y_{1},x_{2},y_{2})$ such that $(x_{1},x_{2})$ and $(y_{1},y_{2})$
are dyadic. For each such dyadic rectangle $R$ that contains at least one point from $S$, we maintain a list of all
points inside it. Our goal will be to allow approximating the longest chain in every $R$.
To this end, we maintain a covering family $\CF(R)$ for the array obtained by writing down the $y$-coordinates
of the points in $R$ in the order of increasing $x$-coordinates. The precise definition and the choice
of parameters for the family is slightly more complex than in the description above, in particular
we need the approximation guarantee to depend on the height of $R$ and be sufficiently good so
that composing $\Oh(\log n)$ approximations still results in the desired bound.
Also, now every cover consists of
a sorted list of intervals, and each interval explicitly stores a chain of appropriate length.

\begin{figure}[h]
\begin{center}
  \includegraphics[scale=0.8]{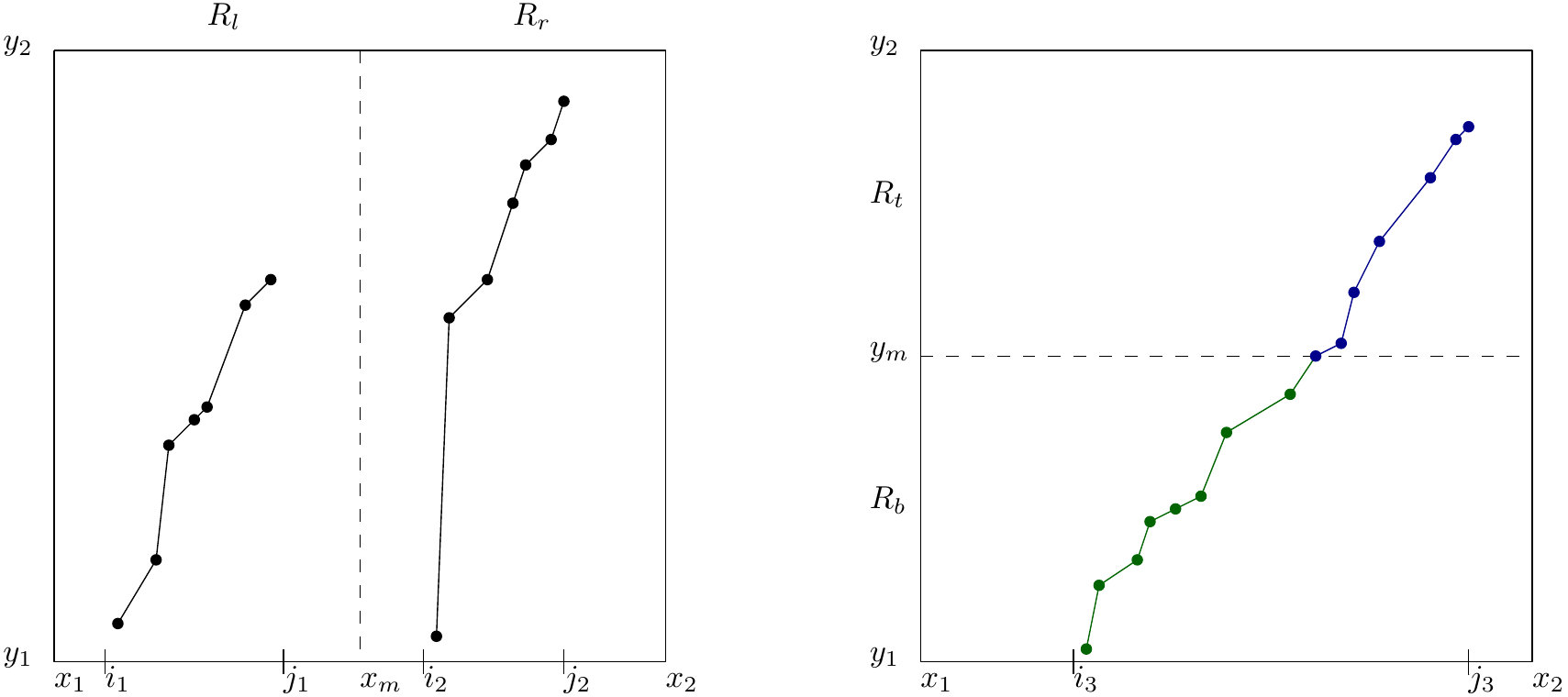}
\end{center}
  \caption{Intervals lying entirely in the left or right rectangle are already covered by their families.
  Intervals crossing $x_m$ will be covered by concatenating pairs of segments from covering families of
  the bottom and top rectangle.}
  \label{fig:fourrect}
\end{figure}

Every point of $S$ belongs to $\Oh(\log^{2}n)$ rectangles, and upon a deletion we need to update
their covers at possibly all the levels. Due to the recursive nature of dyadic rectangles, a cover at
level $k$ of $\CF(R)$ can be obtained as follows. First, we split $R=(x_{1},y_{1},x_{2},y_{2})$ vertically
into $R_{l}=(x_{1},y_{1},x_{m}-1,y_{2})$ and $R_{r}=(x_{m},y_{1},x_{2},y_{2})$, where $x_m=\lceil (x_1+x_2)/2 \rceil$.
We take the unions of covers at levels $k$ in $\CF(R_{l})$ and $\CF(R_{r})$ and observe that we only
need to additionally take care of the queries concerning intervals $(x'_{1},x'_{2})$ with $x'_{1} < x_{m} \leq x'_{2}$.
It turns out that this can be done by splitting $R$ horizontally
into $R_{b}=(x_{1},y_{1},x_{2},y_{m}-1)$ and $R_{t}=(x_{1},y_{m},x_{2},y_{2})$, where $y_m=\lceil (y_1+y_2)/2 \rceil$,
and operating on a number of chains stored in $\CF(R_{b})$ and $\CF(R_{t})$. By inspecting a few cases,
this number can be bounded by the depth of every covering family, which will be kept $\Oh(\eps^{-1}\log n)$.
Roughly speaking, we need to concatenate some chains from appropriately chosen levels of $\CF(R_{b})$ and $\CF(R_{t})$.
Consult Figure~\ref{fig:fourrect}.
Now it is tempting to form level $k$ of $\CF(R)$ by taking the union of levels $k$ from $\CF(R_{l})$
and $\CF(R_{r})$, and adding a small number of new intervals together with their chains.
This is however not so simple, as we would increase the depth of the maintained covering family
in every step of this process, and this could (multiplicatively) accumulate $\log n$ times.
Fortunately, we can run our sparsification procedure on all points belonging to the new chains
to guarantee that the depth remains small.

Taking the union of levels $k$ from $\CF(R_{l})$ and $\CF(R_{r})$ can be done in only $\Oh(\log n)$
by maintaining each cover in a persistent BST. However, running time of the sparsification step
is actually quite large on later levels, and this seems problematic.
We overcome this hurdle by running the sparsification only when sufficiently many deletions have been made
in $R$ to significantly affect the approximation guarantee provided by the cover.
By appropriately adjusting the parameters, this turns out to be
enough to guarantee polylogarithmic update time.

Having obtained a decremental version of our structure, we move to the fully dynamic version.
Now the main issue is that, as insertions can happen anywhere, we cannot work with a fixed collection
of dyadic rectangles. Therefore, we instead apply a two-dimensional recursion resembling 2D
range trees. Then, we need to carefully revisit all the steps of the previous reasoning.
The last step of our construction is removing amortisation. This follows by quite standard
method of maintaining two copies of every structure, the first is used to answer queries while the
second is being constructed in the background. However, this needs to be done in three
places: sparsification procedure, secondary recursion, and primary recursion.

\paragraph{Organisation of the paper. }
We start with preliminaries in Section~\ref{sec:prel}.
Then, in Sections~\ref{sec:cover}, \ref{sec:decremental} and \ref{sec:analysis}
we describe and analyze a decremental structure with deletions working in amortised polylogarithmic time.
First, in Section~\ref{sec:cover} we define the subproblems considered in our structure, introduce the notion of covers,
and explain how to compute covers with good properties at the expense of increasing the approximation guarantee.
We will refer to this technique as sparsification.
Second, in Section~\ref{sec:decremental} we describe how to maintain the information associated
with every subproblem under deletions of points.
Third, in Section~\ref{sec:analysis} we analyse approximation guarantee and running time of the decremental structure.
In Section~\ref{sec:partition} we show to use it to obtain an improved algorithm for Erdős-Szekeres partitioning,
and also provide the details of the $\Omega(n\log n)$ lower bound for this problem in the comparison-based
model.
Finally, in Section~\ref{sec:dynamic} we provide the necessary modifications to make
the structure fully dynamic, and the update time worst-case.

\section{Preliminaries}
\label{sec:prel}

Let $[k]$ denote $\{0,1,\ldots,k-1\}$. A pair of numbers $p_1=(x_1,y_1)$ is smaller than $p_2=(x_2,y_2)$,
denoted $p_1 \prec p_2$, if $x_1 < x_2$ and $y_1 < y_2$.
We usually treat an array of distinct numbers $A=(a_0,a_1,\ldots,a_{n-1})$
as a set $S$ of 2D points with pairwise distinct $x$- and $y$-coordinates
by defining $S=\{(i,a_i) : i \in [n]\}$.
For an array $A$, by its subarray $(i,j)$ we mean the subarray $(a_i,\ldots,a_j)$.
For a set $P$ of points,
by its interval $(i,j)$ we mean the set $P_{i,j}=\{p=(x,y) : p \in P, i \leq x \leq j\}$.
In this work, intervals are always closed on both sides.

For any set of points $P$, its ordered subset of points $X=(x_1,x_2,\ldots,x_k)$ is a chain of length $k$
if $x_1 \prec x_2 \prec \ldots \prec x_k$.
We write $X_{i,j}$ to denote $(x_i,\ldots,x_j)$, and $X_{i}$ refers to $x_{i}$.
$X\circ Y$ denotes the concatenation of chains $X=(x_1,\ldots,x_i)$ and $Y=(y_1,\ldots,y_j)$,
defined when $x_i \prec y_1$.

We want to maintain a structure allowing querying for approximate LIS
under deletions and insertions of elements in an array (called the main array).
The structure returns the approximated length and on demand it can also provide a chain of such length.
An algorithm is $p$-approximate (or provides $p$-approximation), for $p \geq 1$,
if it returns a chain of length at least $r$ when the longest chain has length $p\cdot r$.

A \emph{segment} is a triple $(b,e,L)$, where $b$ and $e$ denote the beginning and the end of a subarray
or an interval, and $L$ is a list of consecutive elements of a chain inside that interval.
The \emph{score} of a segment is the length of $L$.
We use the natural notation for positions of points, subarrays, intervals, or segments.
Point $(k,y)$ is inside interval $(i,j)$ if $i \leq k \leq j$.
Interval $A=(i_1,j_1)$ is to the left of interval $B=(i_2,j_2)$, denoted $A \prec B$,
if $i_1<i_2 \land j_1<j_2$ (note that $A$ and $B$ can still overlap).
$A$ is inside $B$ if $i_{2} \leq i_{1}$ and $j_1 \leq j_2$.

As a subroutine, we use Fredman's $\Oh(n\log n)$ algorithm for computing LIS in an array.
Recall that it can be used to maintain LIS under appending (or prepending, but not both) elements.
That is, we can use it to incrementally compute LIS of the subarrays $(i,i), (i,i+1), \ldots$ so that
after having computed LIS of the subarray $(i,j)$ we can compute LIS of the subarray
$(i,j+1)$ in $\Oh(\log n)$ time. Similarly, we can incrementally compute LIS of the subarrays $(j,j),(j-1,j),(j-2,j), \ldots$.
When we wish to compute the longest chain in a set of points, we simply arrange
its points in an array, sorted by the $x$-coordinates, and then compute LIS with respect to the $y$-coordinates.

\paragraph{BST.}
We use balanced binary search trees to store sets of elements.
Each tree stores items ordered by their keys and provides operations $\mathsf{Insert}(x)$,
$\mathsf{Delete}(k)$ and $\mathsf{Successor}(k)$.
Operation $\mathsf{Find}(k)$ returns an item with the biggest key not greater than $k$,
and $\mathsf{FindRank}(r)$ returns an item with the key of rank $r$ in the set of keys in a BST.
$\mathsf{Join}(T_1,T_2)$ merges two trees,
provided that the keys of all items in $T_1$ are smaller than the keys of items in $T_2$.
$\mathsf{Split}(k)$ divides a tree into two trees, the first containing items with keys less than $k$
and the second containing the remaining items.
$\mathsf{DeleteInterval}(k_1,k_2)$ deletes all elements with keys between $k_1$ and $k_2$; it can be implemented with two splits and one join.

We need a persistent BST which preserves the previous version of itself when it is modified.
This property is used mostly when we join two trees $T_1$ and $T_2$ into one,
but still want to access the original $T_{1}$ and $T_{2}$.
Additionally, we need to augment the tree by storing the size of each subtree to allow for
efficient $\mathsf{FindRank}$ implementation.
All operations work in $\Oh(\log n)$ time, for a tree storing $n$ items, by using e.g. persistent AVL trees~\cite{AVL,Persistent}.
\newpage
\section{Covers and Sparsification}
\label{sec:cover}

In this and several next sections we assume that only deletions and queries are allowed.
Later, we will explain how to modify an algorithm to allow also insertions.
For the case without insertions, we replace each element of the input array by its
rank in the set of all elements.
Thus, the set $S$ representing the input array consists of points with both coordinates from $[n]$.
These coordinates are not renumbered during the execution of the algorithm, we only
delete points from $S$.
For any pair of elements, their current positions in the main array can be compared in 
constant time by checking the $x$-coordinates.
Their values can be also compared in constant time by inspecting the $y$-coordinates.

We say that interval $(a,b)$ is dyadic if $a=i2^k$ and $b=(i+1)2^k-1$
for some natural numbers $i,k$.
Similarly, rectangle $R=(x_1,y_1,x_2,y_2)$ is dyadic if $x_1=i2^{k_1}$, $x_2=(i+1)2^{k_1}-1$,
$y_1=j2^{k_2}$ and $y_2=(j+1)2^{k_2}-1$, for some natural numbers $i,j,k_1,k_2$.
In other words, a dyadic rectangle spans dyadic intervals on both axes.
Rectangle $R$ contains point $s_r=(r,a_r)$ if $x_1 \leq r \leq x_2$ and $y_1 \leq a_r \leq y_2$.
We say that $k_{2}$ is the \emph{height} of a rectangle $R$.
Assume for simplicity that $n$ is a power of $2$.
We consider all dyadic rectangles with $0 \leq x_1,y_1,x_2,y_2 \leq n-1$ possibly containing points from $S$.
These rectangles have height between $0$ and $\log n$.
Moreover, only $\Oh(n\log^2 n)$ of them are nonempty, since each of the $n$ points from $S$
falls into $1+\log n$ dyadic intervals on each of the axes.
Every nonempty rectangle stores the set of points from $S$ inside it in a BST,
and we will think of the rectangles storing just one point as the base case.

Let $1<\lam<2$ be the approximation parameter, where $\lam=1+\eps'$ for some $\eps'$.
Our solution will be able to provide, for any dyadic rectangle $R$ of height $h$, 
$\lam^{4h}$-approximation of the longest chain in $R$.
Thus, the approximation factor in the rectangle encompassing the whole set $S$ will be $\lam^{4\log n}$.
By setting $\eps' = \eps/(8\log n)$, the solution provides an $(1+\eps)$-approximation of LIS,
for any $\eps\in (0,1]$, as $(1+\eps/(8\log n))^{4\log n} < e^{\eps/2} < 1+\eps$.
Furthermore, we have $\log_{\lam}{n}=\log{n}/\log{(1+\eps/(8\log n))}=\Oh(\eps^{-1}\log^2{n})$,
for any $\eps\in (0,1]$.

\paragraph{Covering family.}
In our solution, the approximation will be ensured by storing a sequence of sets of segments,
each consecutive set containing segments with scores larger by a factor of $\lam^2$.
Given an interval of points containing a chain of length $k$, we should be able to find
among stored segments one inside the given interval and with a score close to $k$.
To achieve this, we need the notion of covers and covering families.

\begin{definition}
Consider a set $P$ of points.
\emph{$(k_1,k_2)$-cover} of $P$ is a set $S$ of segments, each with a score of at least $k_1$ and at most $k_2$,
such that for any $i,j$, if $P_{i,j}$ contains a chain of length $k_2$,
then $S$ contains at least one segment $X$ inside interval $(i,j)$,
and we say that $(i,j)$ is \emph{covered} by $X$.
Moreover, we demand that in $S$ no segment is inside another,
so they can be sorted increasingly by their beginnings and ends at the same time and stored in a BST.
\end{definition}

We will refer to a $(k,k)$-cover simply as a $k$-cover, and say that the score of a $(k_1,k_2)$-cover is $k_1$.
If we were operating on real numbers, $\gamma^2$-approximation could be achieved
by storing a sequence of $\Theta(\log_{\gamma}{n})$ covers, namely a
$(\gamma^r,\gamma^{r+1})$-cover for every $0 \leq r < \log_{\gamma}{n}$.
Then, after receiving a query about the longest chain in any interval, we could perform a binary search over the covers,
in order to find the one with the largest score containing a segment inside the given interval.
As unfortunately lengths of chains are natural numbers only, we need a slightly
more complex choice of which covers to store.
Additionally, we need the approximation factor to depend on the height of a rectangle.

\begin{definition}
Let $1<\gamma<2$, $r$ be a natural number greater than $1$, and $P$ a set of $n$ points.
\emph{$(\gamma,r)$-covering family} of $P$ consists of $\Oh(\log_{\gamma} n)$ covers,
ordered by increasing scores.
We call each of these covers a level.
For the first $k=\min(n,3(\gamma-1)^{-1})$ levels, the $i$-th level is just a $i$-cover of $P$.
Then, level $k+j$ is a $(\lceil k\gamma^j \rceil, k\gamma^{j+r-1})$-cover of $P$,
for all $j>0$ such that $k\gamma^{j+r-1}\leq n$.
\end{definition}

\begin{lemma}\label{lem:str}
$(\gamma,r)$-covering family of $P$ provides a $\gamma^r$-approximation of the longest chain
for any interval, in time $\Oh(\log{(\log_{\gamma} n)} \cdot \log n)$.
\end{lemma}
\begin{proof}
We execute a binary search over the levels (sorted by scores).
On each level, we query the BST storing the segments, trying to find a chain that lies inside the given interval $I$.
We stress that the property of admitting such a chain is not monotone over the levels, but
the binary search is still correct because of the following argument.
Choose the largest level $i$ corresponding to a $(k_{1},k_{2})$-cover such that
$I$ contains a chain of length $k_{2}$.
Then, by definition we are guaranteed that on all levels from $1$ to $i$, there is a segment inside $I$.
Possibly, there are larger levels containing segments inside $I$, but this can only help in our
binary search over the levels, and it will always return level $i$ or possibly larger.

There are clearly $\Oh(\log_{\gamma} n)$ levels, each containing at most $n$ segments, so the procedure
takes $\Oh(\log(\log_{\gamma}n)\cdot\log n)$ time.
Suppose the longest chain inside $I$ has length $l$.
If $l \leq k$, the covering family returns the exact answer, as level $l$ is an $l$-cover.
Otherwise, let $e$ be such that $k\gamma^{e} \leq l < k\gamma^{e+1}$.
If $e < r$, then a segment with a score of $k$ from the $k$-th level of the family is good enough.
If $e \geq r$, then the covering family returns a segment with a score of at least
$\lceil k\gamma^{e-(r-1)} \rceil \geq k\gamma^{(e+1)-r}$, which is a $\gamma^r$-approximation.
\end{proof}

Additionally, by a $(\cdot,0)$-covering family of $m$ points we mean a set of $k$-covers,
for each $1\leq k \leq m$, which can always provide an exact answer.

Now we describe a greedy algorithm for creating covers with some additional properties.
Given a cover $C$ consisting of segments $(b_{i},e_{i},L_{i})$, we define \emph{depth} of $C$ as the largest
subset of pairwise intersecting intervals $[b_{i},e_{i}]$. In other words, we calculate the largest
number of such intervals containing the same $x$.
We design two algorithms for creating covers with small depth.
The first one is actually exact, as for the short chains we cannot afford to
decrease their length even by one during our recursive approximation.
The second variant computes sets of segments for further levels of a covering family,
approximating the long chains.

\subsection{Exact cover}
We start with the non-approximate solution.
The greedy algorithm for the exact $k$-cover works as follows.
Assume $k>1$, as $1$-cover is basically a set of all elements.
Starting with the whole input array and an empty cover $C$,
the algorithm finds the shortest prefix $P$ of an array in which there is a chain of length $k$.
Now, some suffixes of $P$ need to be covered, so the algorithm finds the shortest suffix $S$ of $P$
in which there still is a chain of length $k$, as $S$ covers all of the mentioned suffixes.
$S$ and an arbitrary chain of length $k$ inside it is added as a segment to a cover $C$.
Note that both the first and the last element of $S$ must be in any chain of length $k$ in $S$.
Then, elements from the first one in the array up to the first one of $S$ (including them)
are cut off from the array.
This is one step of the algorithm, which then repeats itself.
Pseudocode is presented in Algorithm~\ref{alg:greedy_exact}.

\begin{algorithm}[h]
\begin{algorithmic}[1]
  \Function{Cover-exact}{$A,k$}
  \State \textbf{Input:} array $A = a_0,\ldots,a_{n-1}$, integer parameter $k$.
  \State \textbf{Output:} a collection of segments forming $k$-cover of $A$.
  \Statex
  \State $C \gets \emptyset$
  \State $i \gets 0$
  \While{$i < n$}
    \State $j \gets i$
    \While{$|\mathrm{LIS}(i,j)| < k$ and $j<n$} \Comment{LIS is computed incrementally} \label{LIS1}
      \State $j \gets j+1$
    \EndWhile
    \If{$j=n$}
    \State \Return $C$
    \EndIf
    \State $\triangleright$ $(i,j)$ is the shortest prefix with a chain of length $k$
    \State $q \gets j$
    \While{$|\mathrm{LIS}(q,j)| < k$} \Comment{LIS is computed incrementally} \label{LIS2}
      \State $q \gets q-1$
    \EndWhile
    \State $\triangleright$ $(q,j)$ is the shortest suffix with a chain of length $k$
    \State $C \gets C \cup \{(q,j,\mathrm{LIS}(q,j)\}$, \Comment{$\mathrm{LIS}(q,j)$ is a list of the elements of LIS in $(a_q,\ldots,a_j)$}
    \State $i \gets q+1$ \label{LIS3}
  \EndWhile
  \State \Return $C$
  \EndFunction
\end{algorithmic}
\caption{The greedy algorithm computing an exact cover.}
\label{alg:greedy_exact}
\end{algorithm}

In the remaining part of this subsection we will prove the following theorem.

\begin{theorem}
\label{greedy_exact_1}
Algorithm~\ref{alg:greedy_exact} returns a $k$-cover of depth at most $k$ in $\Oh(nk\log n)$ time.
\end{theorem}

Note that such depth is optimal for a $k$-cover, for example in the case of a sorted array.
We will prove Theorem~\ref{greedy_exact_1} by induction, but first, some additional notation is needed.
Let $P=(a,c)$ be some shortest prefix computed during the execution of the algorithm.
Then, $S=(b,c)$ is its shortest suffix still containing a chain of length $k$,
and let $X=(x_1,\ldots,x_k)$ be the lexicographically minimal (according to the positions of the elements) such chain.
Till the end of this section, we will denote by $x_i$ just a position of an element,
while $|x_i|$ denotes its value.
Recall that $X_{i,j}$ refers to the chain $(x_i,\ldots,x_{j})$.
Let $P'=(b+1,d)$ be the next shortest prefix computed by the algorithm after $P$,
and $X'=(x'_1,\ldots,x'_k)$ be the lexicographically minimal chain in the computed shortest suffix of $P'$.
Our goal is to show the following.

\begin{lemma}\label{greedy_exact_2}
For any two consecutive steps of Algorithm~\ref{alg:greedy_exact} and all $1 \leq i < k$,
we have $x'_i \geq x_{i+1}$.
\end{lemma}

\begin{proof}
The overall strategy is to apply induction on increasing $i$, however inside the inductive step we will need another induction.

Suppose $x'_i < x_{i+1}$.
From the induction hypothesis, or in the base case of $i=1$, we also have $x'_i > x_{i}$.
It cannot be that $|x'_i| < |x_{i+1}|$, because then $X'_{1,i}\circ X_{i+1,k}$
is a chain of length $k$ with $x'_1 > x_1$, so suffix $S$ would not be the shortest one.
Now, if $i=k-1$, then $X_{1,k-1}\circ x'_{k-1}$ is a chain of length $k$ contradicting the minimality of $P$,
so assume $i<k-1$.
It is also not possible that $|x_{i+1}| \leq |x'_i| < |x_{i+2}|$, because then chain
$X_{1,i}\circ x'_i\circ X_{i+2,k}$ is lexicographically smaller than $X$.
Thus, we have $|x'_i| \geq |x_{i+2}|$.

Now, we claim that by secondary induction those properties extends further,
namely for every $i \leq j < k$ we have $x'_j < x_{j+1}$,
and for every $i \leq j < k-1$ we have $|x'_j| \geq |x_{j+2}|$.
The base case of $j=i$ was just proved.
Now, assume the claim holds for $j-1$, so $|x'_j| > |x'_{j-1}| \geq |x_{j+1}|$.
If $x'_j > x_{j+1}$, then $X_{2,j+1}\circ X'_{j,k}$ is a chain of length $k+1$
which contradicts $P'$ being the shortest prefix.
Additionally, $x'_j = x_{j+1}$ cannot hold since $|x'_j| > |x_{j+1}|$.
Thus, we have $x'_j < x_{j+1}$.
But then, for $j < k-1$, it cannot be that $|x'_j| < |x_{j+2}|$, because in such case
$X_{1,i}\circ X'_{i,j}\circ X_{j+2,k}$ is a chain of length $k$
lexicographically smaller than $X$ and inside $S$, since we assumed $x'_i<x_{i+1}$.
Therefore, indeed for every $i \leq j < k$ we have $x'_j < x_{j+1}$.

But then in particular $x'_{k-1} < x_k$, which is the final contradiction needed for the main induction,
since then $X_{1,i}\circ X'_{i,k-1}$ is a chain of length $k$ contradicting minimality of $P$.
\end{proof}

Lemma~\ref{greedy_exact_2} is enough to prove Theorem~\ref{greedy_exact_1} by transitiveness as follows.
During the execution of the greedy algorithm, the first element of any newly computed chain
(which is the first element of the respective shortest suffix) is at position equal or larger
than the position of the last element of a chain computed $k-1$ steps before (which is the last element of the respective suffix).
Therefore, at most $k$ of the computed segments admit nonempty pairwise intersections.
The bound on the time complexity of the whole algorithm follows, as the overall length
of all arrays on which we run Fredman's algorithm is $\Oh(nk)$, and this algorithm takes
logarithmic time per element.
It is easy to see that no computed segment is inside another,
thus all segments can be stored as a BST to meet the definition of a cover.

\subsection{Approximate cover}
The second algorithm uses two integer parameters, $k_1$ and $k_2$ with $k_2 > k_1$,
and its goal is to create a $(k_1,k_2)$-cover.
Let $\Delta = k_2-k_1$.
As it turns out, a greedy solution provides a roughly optimal depth of $k_1/\Delta$.
Similarly as with the exact cover, a greedy algorithm finds the shortest prefix $P$ in which there is a chain of length $k_2$,
then it finds the shortest suffix $S$ of $P$ in which there is a chain of length $k_1$,
adds a segment corresponding to $S$ to a cover $C$ and then repeats,
starting from the first element of $S$.
The pseudocode of this procedure is very similar to the one of Algorithm~\ref{alg:greedy_exact}.
We just modify step~\ref{LIS1} to search for the shortest prefix of length $k_2$ and
step~\ref{LIS2} to search for the shortest suffix of length $k_1$.
Additionally, we change step~\ref{LIS3} to $i \gets q$, without incrementing by 1,
as this is not necessary here and simplifies the proof.
The modified pseudocode is presented in the appendix as Algorithm~\ref{alg:greedy_approx}.
See Figure~\ref{fig:greedy} for a simple example.

\begin{figure}[h]
\begin{center}
  \includegraphics[scale=1.2]{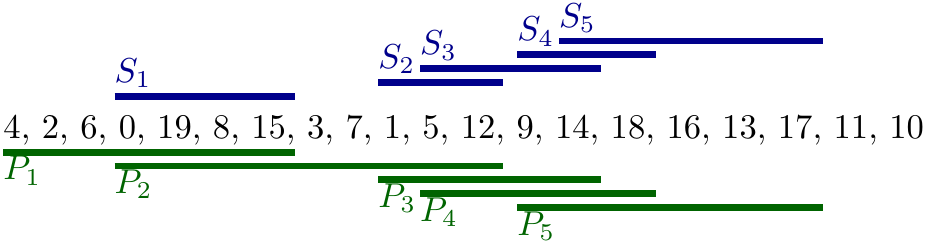}
\end{center}
  \caption{An example of the greedy algorithm for (3,4)-cover.
  The lexicographically minimal LIS in prefix $P_2$ is $(0,3,7,12)$,
  and $(14,16,17)$ in suffix $S_5$.
  The set $S=\{S_1,\ldots,S_5\}$ is enough to cover any subarray containing a chain of length 4,
  for example subarray from value 15 to 14 is covered by $S_2$ and $S_3$.
  The depth of $S$ is $3$, as $S_3$, $S_4$ and $S_5$ have pairwise nonempty intersections.}
  \label{fig:greedy}
\end{figure}

In the remaining part of this subsection we will prove the following theorem.

\begin{theorem}
\label{greedy_approx_1}
Algorithm~\ref{alg:greedy_approx} returns a $(k_1,k_2)$-cover of depth at most $k_1/\Delta$, where $\Delta=k_2-k_1$,
in $\Oh((n\log n) k_1/\Delta)$ time.
\end{theorem}

In order to prove the above theorem, we will again examine the relationship between the elements
of chains in the consecutive computed subarrays.
We will do this in two steps.
Let $P=(a,c)$ be some shortest prefix computed during the execution of the algorithm,
and $X=(x_1,\ldots,x_{k_2})$ be the lexicographically minimal (according to the positions of the elements) chain of length $k_2$ in $P$.
Then, $S=(b,c)$ is the shortest suffix of $P$ still containing a chain of length $k_1$,
and let $Y=(y_1,\ldots,y_{k_1})$ be lexicographically minimal such chain.
Firstly, we need to inductively show the following.

\begin{lemma}\label{greedy_approx_2}
For any two consecutive steps of Algorithm~\ref{alg:greedy_approx} and all $1 \leq i \leq k_1$,
we have $y_i \geq x_{i+\Delta}$.
\end{lemma}
\begin{proof}
The base case of $i=1$ holds, since $X_{1+\Delta,k_2}$ is a chain of length $k_1$.
Now, assume the claim holds for $i-1$.
From the induction hypothesis, we have $y_i > x_{i-1+\Delta}$.
Suppose $y_i < x_{i+\Delta}$.
It cannot be that $|y_i| < |x_{i+\Delta}|$, since then $Y_{1,i}\circ X_{i+\Delta,k_2}$
is a chain of length $k_1+1$, which contradicts the minimality of suffix $S$.
Thus, $|y_i| \geq |x_{i+\Delta}|$ must hold.
But this is a contradiction, because then $X_{1,i-1+\Delta}\circ Y_{i,k_1}$
is a chain of length $k_2$ lexicographically smaller than $X$.
\end{proof}

With the above, we established a connection between (lexicographically minimal) chains in $P$ and $S$,
and now we need to connect the chain in $S$ with the chain in the next prefix $P'$.
So, let $P'=(b,d)$ be the next prefix computed by the algorithm after $P$,
and $X'=(x'_1,\ldots,x'_{k_2})$ is the lexicographically minimal chain of length $k_2$ in $P'$.
Additionally, let $z$ be the maximum number such that $x'_z \leq c$, so exactly $z$ elements of $X'$ are inside $S$.
Note that $z \leq k_1$.

\begin{lemma}\label{greedy_approx_3}
For any two consecutive steps of Algorithm~\ref{alg:greedy_approx} and all $1 \leq i \leq k_1$,
we have $x'_i \geq y_i$.
\end{lemma}
\begin{proof}
For $i>z$ this is obvious.
We proceed by induction on decreasing $i$ with the base case $i=z$.
The base and the step are proved similarly, so now assume that $i\leq z$ and the claim holds for $i+1$.
Suppose $x'_i < y_i$.
It cannot be that $|x'_i| < |y_i|$, because then $X'_{1,i}\circ Y_{i,k_1}$
is a chain of length $k_1+1$ inside $S$.
Thus, we have $|x'_i| \geq |y_i|$.
Now, there are two cases.
If $Y_{1,i}$ is lexicographically smaller than $X'_{1,i}$, we have a contradiction,
since then $Y_{1,i}\circ X'_{i+1,k_2}$ is a chain of length $k_2$ lexicographically smaller
than $X'$ and in $P'$.
To deal with the case of $Y_{1,i}$ being lexicographically larger than $X'_{1,i}$, we first 
need to argue that if $i<k_1$, then $|x'_i| < |y_{i+1}|$.
If $x'_{i+1}=y_{i+1}$ this is trivial, and if not the inequality must hold because otherwise
$Y_{1,i+1}\circ X'_{i+1,k_2}$ is a chain of length $k_2+1$ inside $P'$. 
Having established that $|x'_i| < |y_{i+1}|$ if $i<k_1$, we see that if $Y_{1,i}$ is lexicographically larger than $X'_{1,i}$,
then we also arrive at a contradiction, because $X'_{1,i}\circ Y_{i+1,k_1}$
would be a chain of length $k_1$ lexicographically smaller than $Y$.
Thus, in either case it cannot be that $x'_i < y_i$, so the claim holds.
\end{proof}

Similarly to Theorem~\ref{greedy_exact_1}, Theorem~\ref{greedy_approx_1} holds by transitiveness
and Lemmas~\ref{greedy_approx_2}~and~\ref{greedy_approx_3}, as the positions
of the first elements in two consecutively computed segments shift by at least $\Delta$.
The bound on the time complexity follows by the same argument.

\section{Decremental Structure}
\label{sec:decremental}

Having described covers and how to compute them,
we can go back to our recursive structure of dyadic rectangles, and explain how to maintain their associated information.
Let us consider a rectangle $R=(x_1,y_1,x_2,y_2)$ of height $h$ and the set $P(R)$ of points
from the main array inside $R$.
We will maintain a $(\lam^2,2h)$-covering family of $P(R)$, denoted $\CF(R)$.
At the very beginning, each covering family is constructed by running the greedy algorithms
for every nonempty rectangle.
To maintain the covers while the elements are being deleted, 
level $k$ of $\CF(R)$ is obtained by taking the union of levels $k$ of
two smaller dyadic rectangles, and adding some extra segments.
The extra segments are recomputed from time to time.
More precisely, for each level of $\CF(R)$ we keep a counter. The counter is initially set to $0$ and we
increase it by $1$ whenever a point inside $R$ is deleted. As soon as the counter is sufficiently large,
we recompute the extra segments.
Clearly, we cannot afford to simply run the greedy algorithm on the whole set of points
still existing in $R$.
Instead, we use covering families from some other two smaller dyadic rectangles.
We stress that level $k$ of $\CF(R)$ is updated whenever a point is deleted
$P(R)$, but the extra segments are recomputed only when sufficiently many deletions have occurred.

In this section, we describe how to maintain a single level of $\CF(R)$,
say it is the level providing a $(k_1,k_2)$-cover.
Let $x_m=\lceil (x_1+x_2)/2 \rceil$ and $y_m= \lceil(y_1+y_2)/2 \rceil$.
In the recursive structure of rectangles, $R$ is divided into the left, right, bottom and top rectangles.
Denote them as $R_l=(x_1,y_1,x_m-1,y_2)$, $R_r=(x_m,y_1,x_2,y_2)$,
$R_b=(x_1,y_1,x_2,y_m-1)$, and $R_t=(x_1,y_m,x_2,y_2)$, respectively;
here we ignore the trivial base case of $x_1=x_2$ or $y_1=y_2$ to avoid clutter.
Observe that as long as the approximation factor is tied to the height of a rectangle,
covering families of $R_l$ or $R_r$ contain segments providing good approximation
for any interval $(i,j)$ with $j<x_m$ or $i \geq x_m$.
In fact, $\CF(R_l)$ and $\CF(R_r)$ both contain a level providing a $(k_1,k_2)$-cover.
The left and right rectangle cannot help in the case of intervals crossing the middle point $x_m$, though,
so we need to somehow add segments covering any interval spanning both the left and right rectangle
and containing a chain of length $k_2$.
To this end, we will use $R_b$ and $R_t$, concatenating pairs of segments from their covers.
Consult Figure~\ref{fig:fourrect}.

\paragraph{$\mathsf{Merge}$ procedure.}
We say that interval $(i,j)$ crosses the middle point $x_m$ if $i < x_m$ and $j\geq x_m$.
Let us focus on some interval $(i,j)$ crossing $x_m$, which contains chain $X$ of length $k_2$.
$X$ can be split into two parts (one of them possibly empty),
the first consisting of elements with $y$-coordinates smaller than $y_m$
and the second consisting of the remaining elements of $X$.
Say $X_p < y_m$ and $X_{p+1} \geq y_m$, then these parts are $X_{1,p}$ and $X_{p+1,k_2}$.
Observe that we have access to some approximation of $X_{1,p}$ in $R_b$ and of $X_{p+1,k_2}$ in $R_t$
in the form of segments from the covering families of the bottom and top rectangle.
Let these segments be $X'_b$ and $X'_t$, respectively, and assume they belong to levels $l$ and $l'$
of $\CF(R_b)$ and $\CF(R_t)$, respectively.
If there are many possible candidates for $X'_b$ or $X'_t$, for the analysis we pick arbitrary ones.
We will describe an operation $\mathsf{Merge}$ that selects and concatenates
the relevant pairs of segments from $\CF(R_b)$ and $\CF(R_t)$ for all intervals $(i,j)$ crossing $x_m$,
effectively finding segments that can be used as such $X'_b$ and $X'_t$.

$\mathsf{Merge}$ consists of two steps.
The first one is creating an initial set of segments.
There are four cases of how $X'_b$ and $X'_t$ can be located.
Consult Figure~\ref{fig:merge}.

\begin{figure}[h]
\begin{center}
  \includegraphics[scale=0.8]{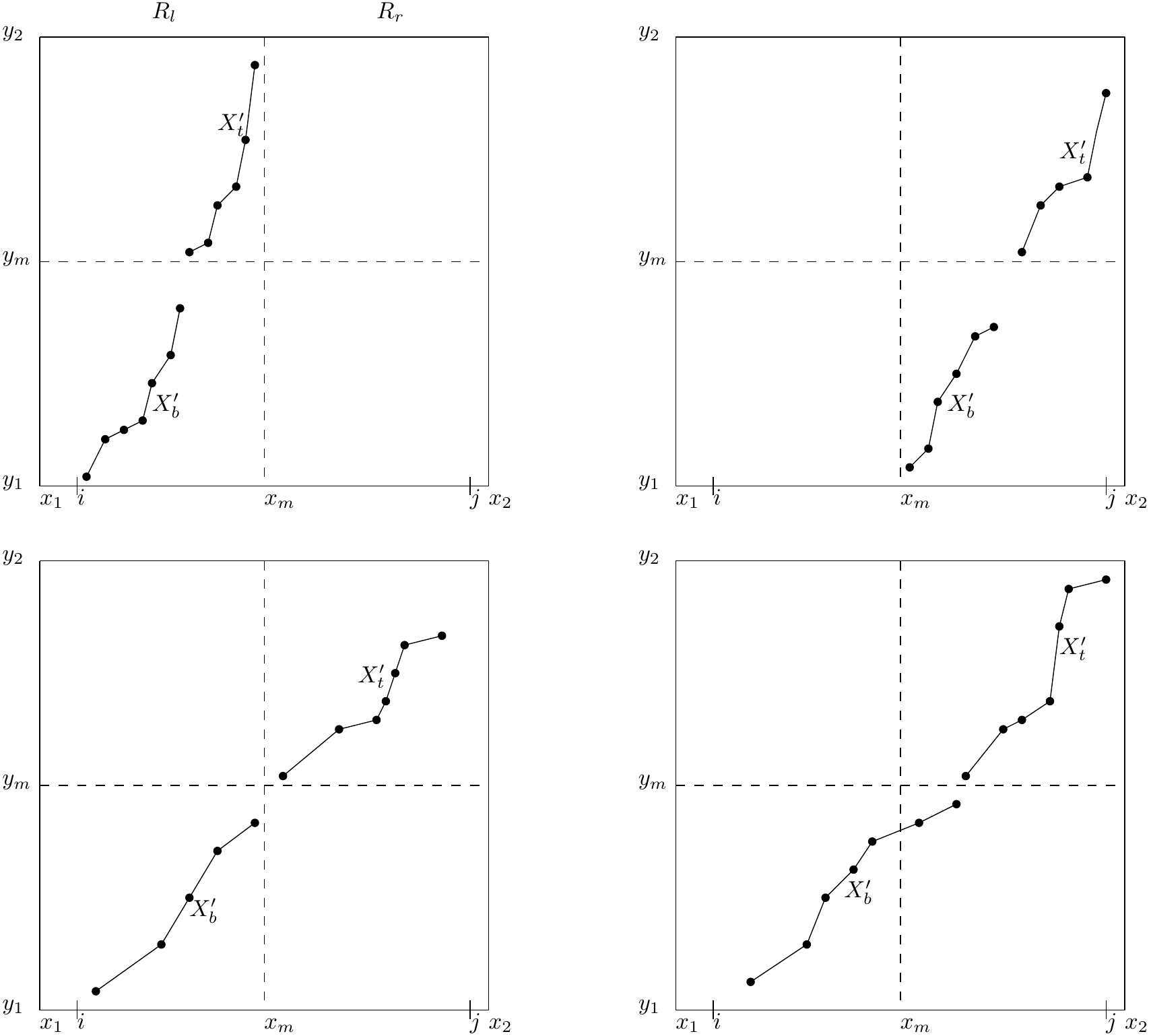}
\end{center}
  \caption{Possible cases of positions of $X'_b$ and $X'_t$.}
  \label{fig:merge}
\end{figure}

\begin{itemize}
\item Both $X'_b$ and $X'_t$ are inside the left rectangle $R_l$.
Recall that segments in each level of a covering family are stored in a BST.
Thus, we can quickly find the rightmost segment on level $l'$ of $\CF(R_t)$
which is inside $R_l$, say it is segment $S'=(a,b,L)$.
Observe that $S' \succeq X'_t$ must hold.
We append $S'$ to the rightmost segment $S$ on $l$-th level of $\CF(R_b)$
whose interval ends before $a$.
The resulting segment is inside $(i,j)$ (as $S \succeq X'_b$ must hold)
and provides a good approximation for a chain in $(i,j)$, as $S$ and $S'$ comes from levels $l$ and $l'$.

Observe that even though both $X'_b$ and $X'_t$ reside in the left rectangle,
we are not guaranteed to obtain the sought approximation of $X$ just from $\CF(R_l)$.
$X'_b\circ X'_t$ is already an approximation,
so taking another approximation of it would square the approximation factor, which we cannot afford there.

\item Both $X'_b$ and $X'_t$ are inside right rectangle $R_r$.
This is symmetric to the first case, we could find the leftmost segment $S$ on level $l$ of
$\CF(R_b)$ which is inside $R_r$, and append to it the leftmost segment 
on level $l'$ of $\CF(R_t)$ starting after $S$.

\item $X'_b$ is inside $R_l$ and $X'_t$ is inside $R_r$.
Then, we can find the rightmost segment $S$ on level $l$ of $\CF(R_b)$ which is inside $R_l$,
and the leftmost segment $S'$ on level $l'$ of $\CF(R_t)$ which is inside $R_r$.
It must be that $S \succeq X'_b$ and $S' \preceq X'_t$, thus $S\circ S'$ provides good approximation.

\item One of $X'_b$, $X'_t$ is not inside either left or right rectangle.
Say that $X'_b$ crosses $x_m$ and $X'_t$ is inside $R_r$, as the other case is analogous.
Here, we can simply iterate over all segments on level $l$ of $\CF(R_b)$ that cross $x_{m}$,
and append to each of them the leftmost segment on level $l'$ of $\CF(R_t)$ starting further.
This allows us to find $X'_{b}$, and if it was concatenated with $S'$ then $S' \preceq X'_t$.
\end{itemize}

There are obviously two issues.
The first issue is that we are guessing levels $l$ and $l'$.
However, we can afford to iterate through all possibilities, as there are not so many levels,
namely $\Oh(\log_{\lam} n) = \Oh(\eps^{-1}\log^2 n)$.
The second issue is that we are iterating through every segment crossing $x_m$
and we do not want this quantity to be big.
However, this is guaranteed by maintaining covers of small depth.

To sum up, performing the first step of $\mathsf{Merge}$ as described above gives us the desired approximation,
assuming $\CF(R_b)$ and $\CF(R_t)$ provide a good approximation.
The pseudocode of both steps of $\mathsf{Merge}$ is shown as Algorithm~\ref{alg:merge},
ignoring some corner cases.
Namely, notice that 'the rightmost segment that is inside $R_l$' could not exist (and similarly in the symmetric cases),
that is, there might be no segments inside $R_l$ in the cover on that level.
We add artificial guards to each cover, namely segments $(-\infty,-\infty,\emptyset)$ and $(\infty,\infty,\emptyset)$.
The algorithm is aware that those segments are empty and does not concatenate them,
and we will ignore this detail from now on.
Additionally, as one of $X'_b, X'_t$ can be empty, we pretend that each covering family contains an artificial
level with covers of zero score.
Keeping this in mind, the set of segments $U$ computed in line~\ref{candidates} is enough to cover all intervals crossing $x_m$.
However, it is slightly too big.

The goal of the second step of $\mathsf{Merge}$ (in lines~\ref{spars1}-\ref{spars2}) is to decrease
the number of segments to $\Oh(\log n)$ at the expense of worsening the approximation by a factor of $\lam$.
Suppose $U$ provided segments with scores of at least $k_1\lam^2$.
We create a set $P$ of points appearing in a chain of any segment from $U$ (we stress that only the points
belonging to these chains are considered, not all points in the corresponding intervals),
and then compute a $(\lceil k_1\lam \rceil,k_1\lam^2)$-cover $U'$ of $P$ with the greedy algorithm
(for $k_1 \leq 3/\eps'$, we calculate an exact cover).
As $U$ provided a segment with a score of at least $k_1\lam^2$ inside any interval crossing $x_m$ 
and containing a chain of length at least $k_2$, now $U'$ provides a segment with a score of at least $k_1\lam$.

\begin{algorithm}
\begin{algorithmic}[1]
  \Function{Merge}{$F_b,F_t,m,k,apx$}
  \State \textbf{Input:} bottom and top covering families, middle element $m$, integer $k$, Boolean $apx$.
  \State \textbf{Output:} set of segments.
  \Statex
  \State $U \gets \emptyset$
  \For{$C$ being $(r_1,r_2)$-cover in $F_b$} \label{loop1}
  \Comment{Iterate over levels}
  \If{$apx = \texttt{false}$}
    \State Let $C'$ be $(r'_1,r'_2)$-cover in $F_t$ for $r'_1$ such that $r_1+r'_1 = k$
    \Else
    \State Let $C'$ be $(r'_1,r'_2)$-cover in $F_t$ with the smallest $r'_1$ such that $r_1+r'_1 \geq k$
    \EndIf
    
    \Statex
	\State Let $s\in C$ be segment $(s_i,s_j,s_L)$ with the largest $s_j$ such that $s_j < m$
	\State Let $s'\in C$ be segment $(s'_i,s'_j,s'_L)$ with the lowest $s'_i$ such that $s'_i \geq m$
	\State $\triangleright$ Those are two non-crossing segments on the left and right of $m$   
    \For{segment $u=(u_i,u_j,u_L)$ in $C$ between $s$ and $s'$} \label{loop2}
    \Comment{All segments crossing m}
      \State Find segment $u'=(u'_i,u'_j,u'_L)$ in $C'$ with the lowest $u'_i$ such that $u'_i>u_j$
      \State $V \gets (u_L\circ u'_L)_{1,k}$ \Comment{Merging two chains and trimming}
      \State $U \gets U \cup \{(V_1,V_k,V)\}$
    \EndFor
    \Statex
    \State Let $s\in C'$ be segment $(s_i,s_j,s_L)$ with the largest $s_j$ such that $s_j < m$
	\State Let $s'\in C'$ be segment $(s'_i,s'_j,s'_L)$ with the lowest $s'_i$ such that $s'_i \geq m$    
    \For{segment $u'=(u'_i,u'_j,u'_L)$ in $C'$ between $s$ and $s'$} \label{loop3}
      \State Find segment $u=(u_i,u_j,u_L)$ in $C$ with the largest $u_j$ such that $u_j<u'_i$
      \State $V \gets (u_L\circ u'_L)_{1,k}$
      \State $U \gets U \cup \{(V_1,V_k,V)\}$
    \EndFor
  \EndFor
  \State \label{candidates}
  \State $P \gets \{p : p \in L$ for any $(i,j,L)\in U\}$ \label{spars1}
  \State Store $P$ as an array
  \If{$apx = \texttt{false}$}
    \State $U' \gets \Call{Cover-exact}{P,k}$
    \Else
    \State $U' \gets \Call{Cover-approx}{P,\lceil k/\lam \rceil,k}$ \label{spars2}
    \EndIf
  \State In $U'$ keep only segments crossing $m$ and two non-crossing closest to $m$ from both sides
  \State \Return $U'$
  \EndFunction
\end{algorithmic}
\caption{Implementation of $\mathsf{Merge}$.}
\label{alg:merge}
\end{algorithm}

Let us state some properties of $\mathsf{Merge}$ more precisely.
Depending on the last Boolean argument, we have approximated and exact variants.

\begin{lemma} \label{lem:merge1}
Assume covering families of the bottom and top rectangle, denoted $\CF(R_b)$ and $\CF(R_t)$, provide $p$-approximation.
Let $S$ be the set of segments returned after invoking an approximation variant of $\mathsf{Merge}$
with parameters $\CF(R_b),\CF(R_t),x_m,k,\normalfont{\texttt{true}}$.
Then, for any interval $(i,j)$ that crosses $x_m$ and contains a chain $X$ of length $pk$,
$S$ contains a segment with a score of at least $\lceil k/\lam \rceil$ inside $P(R)_{i,j}$.
\end{lemma}
\begin{proof}
This mostly follows from the discussion above.
$X$ can be partitioned into two parts $X_b$ and $X_t$ lying entirely inside $R_b$ and $R_t$, respectively.
$\CF(R_b)$ and $\CF(R_t)$ provide approximations of these parts, denoted $X'_b$ and $X'_t$, with the sum of lengths at least $k$;
in fact, there might be many possible segments $X'_b$ and $X'_t$ covering $X_b$ and $X_t$.
As described before in the four cases, the first part of $\mathsf{Merge}$
finds some $X'_b\circ X'_t$ and adds it to set $U$.
Then in the second part we run the greedy algorithm to create a cover $U'$ for the set $P$ of points appearing
in a chain of any segment of $U$.
For any segment $(i',j',L) \in U$, $P_{i',j'}$ clearly contains a chain of length $k$.
Additionally, for any interval $(i,j)$ in $P$, if $P_{i,j}$ contained a chain of length $k$,
then $U'$ contains a segment with a score of $\lceil k/\lam \rceil$ inside $(i,j)$.
It is easy to see that this remains true for all intervals crossing $x_m$ if we keep in $U'$
only the segments crossing $x_{m}$,
the rightmost segment lying entirely in $R_l$, and the leftmost segment lying in $R_r$.
\end{proof}

A similar lemma holds for levels providing the exact covers.
The proof is analogous but simpler, since there is no approximation at all, and we omit it here.

\begin{lemma} \label{lem:merge2}
Assume in the covering families of the bottom and top rectangle, denoted $\CF(R_b)$ and $\CF(R_t)$, levels up to $3/\eps'$ are the exact covers.
Let $S$ be the set of segments returned after invoking an exact variant of $\mathsf{Merge}$
with parameters $\CF(R_b),\CF(R_t),x_m,k,\normalfont{\texttt{false}}$, with $k \leq 3/\eps'$.
Then, for any interval $(i,j)$ that crosses $x_m$ and contains a chain of length $k$,
$S$ contains a segment with a score of $k$ inside $P(R)_{i,j}$.
\end{lemma}

Due to the application of the greedy algorithm in the second part, the running time of $\mathsf{Merge}$
might be unacceptably high, especially for large values of $k_1$.
Therefore, we will run it only after a number of points has been deleted from $R$, making its running
time amortised at the expense of increasing the approximation by a factor of $\lam$.
Recall that we maintain a counter associated with every level of the covering family, and
increase it whenever a point is deleted from $R$.
For a level providing a $(k_1,k_2)$-cover, the counter goes from $0$ to $\max(1,\lfloor \eps' k_1 \rfloor -1)$
(recall $\lam =1+\eps'$).
Then, we run $\mathsf{Merge}$ to recompute the cover. Because we recompute after having
deleted $\eps' k_{1}$ points from $R$, the score of a segment might decrease from
$k_1\lam$ to $k_1\lam - \eps' k_1 \geq k_1$, which is acceptable,
and the cost of running $\mathsf{Merge}$ is distributed among $\eps' k_1$ deletions.

The case of the first $\Oh(1/\eps')$ levels of a covering family is special, as for them
we cannot afford to increase the approximation, even by adding one.
Therefore, these covers need to be exact, and so we run $\mathsf{Merge}$ after every
deletion of an element from $R$.
The cost of doing so is not too high, as a single $\mathsf{Merge}$ for these levels considers only
a polylogarithmic number of segments and points.

\paragraph{Bound on depth.}
Let us now prove the effects of sparsification as the second step of $\mathsf{Merge}$ on the depth of elements.
Note that only the new segments computed by $\mathsf{Merge}$ can cross $x_m$.

\begin{lemma} \label{lem:Tfreq}
The depth of a set of segments returned by $\mathsf{Merge}$ is $\Oh(\eps^{-1}\log n)$.
\end{lemma}
\begin{proof}
The second step of $\mathsf{Merge}$ is sparsification by the greedy algorithm.
Assume that we consider level $l$ of a covering family.
If $l \leq 3/\eps'$, then it was the greedy algorithm for the exact cover.
As $1/\eps' = \Oh(\eps^{-1}\log n)$ and the maximum depth of the greedy algorithm for $k$-cover is $k$
by Theorem~\ref{greedy_exact_1}, the claim holds.
If $l$ is larger, then the greedy algorithm for $(\lceil k/\lam \rceil, k)$-cover was used.
In that case, by Theorem~\ref{greedy_approx_1} the maximum depth is
\begin{align*}
\frac{\lceil k/\lam \rceil}{\Delta} &= \frac{\lceil k/\lam \rceil} {k-\lceil k/\lam \rceil}
< \frac{k/\lam + 1}{k-k/\lam-1} = \frac{k+\lam}{k\lam-k-\lam} < \frac{k+2}{k\lam-k-2}= \\&= \frac{k+2}{\eps' k -2} = \frac{1+2/k}{\eps'-2/k} < \frac{2}{\eps'-2\eps'/3} = \frac{6}{\eps'}= \Oh(\eps^{-1}\log n),
\end{align*}
as $1<\lam<2$ and $k>3/\eps'$.
\end{proof}

Notice that single sparsification does not have a huge impact.
In Algorithm~\ref{alg:merge}, the main loop in line~\ref{loop1} iterates over levels of a covering family,
and there are $\Oh(\log_{\lam} n)=\Oh(\eps^{-1}\log^2 n)$ of them.
Secondary loops in lines~\ref{loop2}~and~\ref{loop3} iterate over segments crossing the middle $x$-coordinate
(and two additional extreme segments), and according to Lemma~\ref{lem:Tfreq} there are $\Oh(\eps^{-1}\log n)$ such segments.
Together this makes at most $\Oh(\eps^{-2}\log^3 n)$ segments produced in the first step of $\mathsf{Merge}$,
before sparsification reduces their number again to $\Oh(\eps^{-1}\log n)$ segments crossing $x_m$.
In one step it may not seem significant, but it should be evident that this reduction in the number of segments
is crucial, as otherwise we basically increase the depth multiplicatively by the number of levels
for each increase in height of a rectangle.

\paragraph{Maintaining a cover as BST.}
After all steps of $\mathsf{Merge}$, we have a set of segments constructed with help of $\CF(R_b)$ and $\CF(R_t)$,
but we still need to integrate it with covers of $R_l$ and $R_r$ into a single cover.
To this end, we need to ensure that no segment is inside some other segment,
and then put all of them into one BST.
The latter is easy, as there are few new segments.
We can just join trees of the left and right rectangle, then insert each new segment into the resulting tree.
All operations here are done on persistent BSTs.
But first, we need the non-inclusion property.
Obviously, segments in trees from $R_l$ and $R_r$ are disjoint,
so only the new ones computed by $\mathsf{Merge}$ can pose a problem.

Observe that from the new segments $(i,j,L)$ with $j<x_m$, so being entirely in the left rectangle,
we can keep only the rightmost one (or they might be no such segments at all).
Similarly, from new segments lying entirely inside the right rectangle,
we can keep only the leftmost one.
This is done in the last lines of Algorithm~\ref{alg:merge}.
As long as one of those two segments is inside some segment from the left or right tree,
we delete that unnecessary segment with a larger interval from its BST.
If more than one segment needs to be deleted from one tree, they are continuous in that BST,
so we can detect them all and use $\mathsf{DeleteInterval}$.
Now, we have these two segments and segments crossing the middle point, so
$\Oh(\eps^{-1}\log n)$ new segments, and none of them can lie inside another.
We only need to check whether any of the segments from the left or right tree lies inside
any of those new segments and if so, we delete any such unnecessary new segment with a larger interval.
At the end, we join the left and right tree and insert all the remaining new segments into
the resulting tree, one by one.
The final tree is stored as a cover, with its counter reset to zero.

In fact, joining trees from the left and right rectangle and then inserting segments returned
by $\mathsf{Merge}$ is done not only on recomputing a cover on some level,
but after \emph{every} deletion of a point from $P(R)$.
$\CF(R)$ relies on covers from smaller rectangles, so whenever they change, that should be taken into account.
This is not too costly, though, as a single point is inside $\Oh(\log^2{n})$ rectangles,
and the complexity of joining is polylogarithmic regardless of the size of the rectangle.
Thus, each level of $\CF(R)$ stores $\Oh(\eps^{-1}\log n)$ segments returned by the last $\mathsf{Merge}$
for that level.
An exception is when no $\mathsf{Merge}$ occurred yet, that is when a cover on that level
still comes from the preprocessing.
Then no segments are stored and no trees are joined.

\section{Analysis of the Decremental Structure}
\label{sec:analysis}

In this section, we summarize how operations on our recursive structure of rectangles work,
and analyse the approximation factor and the running time.
First, let us recall what data is stored during the execution of the algorithm.

\paragraph{Information stored by the algorithm.}
On the global level, we store a pointer to the biggest rectangle.
Each nonempty rectangle $R$ stores:
\begin{itemize}
\item At most four pointers to $R_l,R_r,R_b,R_t$, whenever they are nonempty.
\item Set $P(R)$ of points inside $R$, sorted by $x$-coordinates and stored in a BST. 
\item Covering family $\CF(R)$ of $P(R)$.
      $\CF(R)$ is an array of covers.
      Each cover is a BST of segments, ordered by coordinates (both simultaneously).
      Additionally, there is a counter tied to every cover.
\item For each level of $\CF(R)$, a sorted list of segments returned by the last $\mathsf{Merge}$ on that level.
In case of levels coming from preprocessing and not recomputed yet, nothing is stored here.
\end{itemize}

When an element is deleted, it is immediately deleted only in BSTs storing points inside rectangles.
We do not modify the list stored for every segment in a cover, even if some of its elements are deleted.
Instead, we recompute the cover and create new segments once in a while.
The algorithm also stores a global Boolean array providing in constant time information
which element of the input array has been already deleted.

\paragraph{Querying for LIS.}
Answering a query for approximated LIS inside any subarray of the main array is relatively uncomplicated.
We need to use only the covering family of the main rectangle, which contains all of the elements.
After a binary search over levels, some segment $S$ is retrieved from the level providing $(k_1,k_2)$-cover.
As an answer to a query, we report the lower bound $k_1$,
not the actual length of the list storing elements of a chain in $S$ (it could be bigger, but also possibly contains some deleted elements).
If needed, in additional time $\Oh(k_1)$ the list can be traversed and elements forming the chain reported, excluding the elements that were deleted after forming that list.
Therefore, we have the following.

\begin{lemma} \label{lem:Tfind}
The worst-case complexity of returning the approximated length of LIS in any subarray is 
$\Oh(\log n \cdot \log{(\eps^{-1}\log^2 n}))=\Oh(\log^2 n)$.
If the returned length is $k$, then the algorithm can return the corresponding increasing subsequence in $\Oh(k)$ time.
\end{lemma}

Observe that we store points with $x$-coordinates equal to the initial indices in the input array.
If we wish to answer queries for approximated LIS in subarray $(i,j)$ of the main array,
so when $i,j$ are the current indices of elements, we need some translation.
It is done simply by using $\mathsf{FindRank}(i)$ in a BST of the biggest rectangle,
which stores exactly the set of non-deleted points.
This way, we translate the current index into the initial index.
The same method applies to deleting elements.

\paragraph{Handling a single deletion.}
When an element is deleted, the algorithm needs to go through each rectangle containing the point
with coordinates defined by the initial index and value of that element.
There are $\Oh(\log^2 n)$ such rectangles, and they are visited in order from the smallest to the biggest.
By that, we mean in order of increasing span on $y$-axis, then increasing span on $x$-axis.
For each rectangle, the counters of every level of their covering families are incremented.
Some counters might reach their maximum value, which triggers computing $\mathsf{Merge}$ for the cover on that level.
We also need to delete an element from BSTs storing points inside those rectangles.
Finally, this element is marked as deleted in the global Boolean array.

\paragraph{Approximation factor.} As mentioned earlier, the approximation guarantee is tied to the height of the rectangle.
Intuitively, sparsification and amortisation are responsible for the approximation factor being $\lam^{d\cdot 2h}$
for some constant $d$.
We use covering families of rectangles $R_b$ and $R_t$ with heights smaller by one,
then sparsification shortens segments by a factor of $\lam$.
Next, we let $\eps' k_1$ deletions to happen before recomputing,
which shortens segments at most by the same factor.
This is formalised in the following lemma.
\begin{lemma} \label{lem:approx}
Rectangle $R$ of height $h$ maintains a $(\lam^2,2h)$-covering family.
\end{lemma}
\begin{proof}
We use induction on the size of the rectangle.
Rectangles of zero height contain at most one element and the claim is trivial.
Let us now consider rectangle $R$ of height $h$, with its left/right/bottom/top rectangles $R_l,R_r,R_b,R_t$.
From the assumption, $R_l$ and $R_r$ provide $(\lam^2,2h)$-covering families,
while $R_b$ and $R_t$ provide $(\lam^2,2(h-1))$-covering families.
The first $k=3/\eps'$ levels in $\CF(R)$ are exact covers, recomputed with $\mathsf{Merge}$ after every deletion,
therefore they are correct by Lemma~\ref{lem:merge2}. 

Let us focus on a level providing $(\lceil k\lam^{2j} \rceil, k\lam^{2(j+2h-1)})$-cover for some $j>0$.
Segments coming from the same level in $R_l$ and $R_r$ have the correct score.
We need to consider segments coming from $\mathsf{Merge}$,
more precisely invoking $\mathsf{Merge}(\CF(R_b),\CF(R_t),x_m,k\lam^{2j+2}, \normalfont{\texttt{true}}$).
$\CF(R_b)$ and $\CF(R_t)$ provide $\lam^{4h-4}$-approximation.
Say that we have an interval crossing the middle point and containing chain $X$ of length at least $k\lam^{2(j+2h-1)}$.
By Lemma~\ref{lem:merge1}, in a set returned by $\mathsf{Merge}$ we have a chain of length
$\lceil k\lam^{2j+1} \rceil$ covering $X$.

Finally, because of amortising $\mathsf{Merge}$, there might be up to
$\lfloor \eps' \lceil k\lam^{2j} \rceil \rfloor -2$ deletions before recomputing the cover.
Therefore, the guarantee on the length of a chain is
\begin{eqnarray*}
\lceil k\lam^{2j+1} \rceil - \lfloor \eps' \lceil k\lam^{2j} \rceil \rfloor + 2 &\geq& \lam k\lam^{2j} - (\lam-1)\lceil k\lam^{2j} \rceil + 2 \\
&>& \lam \lceil k\lam^{2j} \rceil - (\lam-1)\lceil k\lam^{2j} \rceil \\
& = & \lceil k\lam^{2j} \rceil,
\end{eqnarray*}
where we have used the assumption $\lam < 2$.
\end{proof}

\paragraph{Time complexity. } To prove that the amortised complexity of a deletion
is polylogarithmic, we first need to bound the time complexity of a single $\mathsf{Merge}$.

\begin{lemma} \label{lem:Tmerge}
For any rectangle $R$, recomputing a $(k_1,k_2)$-cover from covering family of $R$ works in
$\Oh(k_1 \eps^{-3}\log^5 n)$ time.
\end{lemma}
\begin{proof}
First, observe that for our construction we have $k_2=\Oh(k_1)$.
As stated before, the first step of $\mathsf{Merge}$ produces $\Oh(\eps^{-2}\log^3 n)$ segments.
This takes time $\Oh((k_1 + \log n)\eps^{-2}\log^3 n )$, as for each segment there is one search in a BST involved and then
we create a list containing elements of a chain.
The score of each segment is $\Oh(k_1)$, so set $P$ of all points appearing in any of the segments
is of size $\Oh(k_1 \eps^{-2}\log^3 n)$.
For the second step of $\mathsf{Merge}$ (sparsification), we sort $P$ and run the greedy algorithm for finding a cover.
This takes $\Oh(k_1\eps^{-3} \log^5 n)$ time by Theorem~\ref{greedy_exact_1}, Theorem~\ref{greedy_approx_1} and Lemma~\ref{lem:Tfreq}.
Now the only thing that remains is joining BSTs representing covers of the left and right rectangle,
then inserting the newly computed segments, all while ensuring there is no segment inside another.
As described before, it is done with $\Oh(\eps^{-1}\log n)$ operations
$\mathsf{Find}$, $\mathsf{Join}$, $\mathsf{Delete}$ and $\mathsf{DeleteInterval}$.
Thus, creating the final BST takes time $\Oh(\eps^{-1}\log^2 n)$.
\end{proof}

\begin{lemma} \label{lem:Tamort}
For any rectangle $R$, the amortised cost of recomputing $(k_1,k_2)$-cover from $\CF(R)$ is
$\Oh(\eps^{-4}\log^6 n)$.
\end{lemma}
\begin{proof}
For the case of $k_1=k_2$, $k_1=\Oh(1/\eps')=\Oh(\eps^{-1}\log n)$ and thus the claim holds.
Otherwise, the step of recomputing that level of $\CF(R)$ is amortised between $\lfloor \eps k_1 \rfloor - 1$ deletions.
We have $\eps' k_1 \geq 3$, and so $\lfloor \eps' k_1 \rfloor - 1 \geq 2\eps' k_1 /3$.
Therefore, using Lemma~\ref{lem:Tmerge}, the amortised cost is $\Oh(\eps^{-3}(\log^5 n) / \eps') = \Oh(\eps^{-4}\log^6 n)$.
\end{proof}

Now we need to sum up the cost of deletion among all the relevant rectangles.

\begin{lemma} \label{lem:Ttotal}
The amortised cost of deletion is $\Oh(\eps^{-5}\log^{10} n)$.
\end{lemma}
\begin{proof}
While deleting an element, the algorithm needs to visit $\Oh(\log^2 n)$ rectangles
containing a point corresponding to that element.
For each such rectangle $R$, we do the following:
\begin{itemize}
\item Delete the point from the BST storing points inside $R$, in time $\Oh(\log n)$.
\item Increment the counter for each level of $\CF(R)$, there are $\Oh(\log_{\lam} n)=\Oh(\eps^{-1}\log^2 n)$ levels.
\item Pay the amortised cost of recomputing every level of $\CF(R)$, by Lemma~\ref{lem:Tamort} this is $\Oh(\eps^{-5}\log^8 n)$.
\item For every level of $\CF(R)$, we join trees from $\CF(R_l),\CF(R_r)$ with segments returned by 
the last $\mathsf{Merge}$.
This takes time $\Oh(\eps^{-1}\log^2 n)$ for each level, so $\Oh(\eps^{-2}\log^4 n)$ in total.
\end{itemize}

Clearly, the third item is dominating.
\end{proof}

The last thing that should be investigated is the time complexity of preprocessing.

\begin{lemma} \label{lem:Tprep}
Preprocessing of the input array takes time $\Oh(n\eps^{-2}\log^6 n)$.
\end{lemma}
\begin{proof}
Recall that each point is inside $\Oh(\log^2 n)$ rectangles.
Thus, the recursive structure of nonempty rectangles can be built in time $\Oh(n \log^2 n)$,
and BSTs of points inside each rectangle are constructed in total time $\Oh(n \log^3 n)$.

The main component is computing every covering family with the greedy algorithm.
For rectangle $R$, the first $\Oh(1/\eps')$ levels of $\CF(R)$ are computed with the exact greedy algorithm,
and then every $(k_1,k_2)$-cover is computed with the approximation greedy algorithm
with parameters $P(R)$,$\lceil k_1\lam \rceil$,$k_2$.
As one point is in $\Oh(\log^2 n)$ rectangles, a covering family of any rectangle has $\Oh(\eps^{-1}\log^2 n)$ levels,
and by computations identical to those in Lemma~\ref{lem:Tfreq} the overhead of running the greedy algorithm for each element is also $\Oh(\eps^{-1}\log^2 n)$,
the total time of computing covering families is $\Oh(n\eps^{-2}\log^6 n)$.
\end{proof}

\section{Erdős-Szekeres Partitioning}
\label{sec:partition}

Having finished the description of our decremental structure, we can state a particularly interesting
application to decomposing a permutation on $n$ elements into $\Oh(\sqrt{n})$ monotone
subsequences in $\tildeOh(n)$ time. By repeatedly applying the Erdős-Szekeres theorem, 
that guarantees the existence of a monotone subsequence of length at least $\sqrt{n}$
in any permutation on $n$ elements, it is clear that such a decomposition exists. Now we
can find it rather efficiently.

\begin{lemma} \label{lem:Decomp}
A permutation on $n$ elements can be partitioned into $\Oh(\sqrt n)$ monotone subsequences in $\tilde{\Oh}(n)$ time.
\end{lemma}
\begin{proof}
We maintain a 2-approximation of LIS with our decremental algorithm. As long as the approximated
length is at least $\sqrt{n}$, we retrieve its corresponding increasing subsequence and delete all of
its elements. This takes polylogarithmic time per deleted elements, so $\tildeOh(n)$ overall, and
creates at most $\sqrt{n}$ increasing subsequences.
Let the length of LIS of the remaining elements be $k$. Because we have maintained a 2-approximation,
$k< 2\sqrt{n}$.
Next, we run Fredman's algorithm on the remaining part of the permutation in $\Oh(n\log n)$ time.
Recall that it scans the elements of the input sequence $(a_{1},a_{2},\ldots,a_{m})$ and computes
for every $i$ the largest $\ell$ such that there exists an increasing subsequence of length $\ell$
ending with $a_{i}$. Then, it is straightforward to verify that, for every $\ell=1,2,\ldots,k$, the elements
for which this quantity is equal to $\ell$ form a decreasing subsequence, so we obtain
a partition of the remaining elements into $k$ decreasing subsequences.
Thus, we partitioned the whole sequence into $\Oh(\sqrt n)$ monotone subsequences.
\end{proof}

We remark that our algorithm can be used to partition a permutation into $(1+\eps)\sqrt{2n}$ monotone
subsequences (for any constant $\eps>0$, still in $\tilde{\Oh}(n)$ time),
where $\sqrt{2n}$ is the optimal number of subsequences~\cite{BrandstadtK86}.
To achieve this, let $s=\sqrt{2n}$.
We iteratively try extracting increasing subsequences, in step $i$ a subsequence of length at least $s-i$,
for $i=0,1,\ldots$.
This process must terminate before $i=s-1$, as $s+(s-1)+(s-2)+...+2 > n$.
When the process terminates, we have that the length of LIS is less than $(1+\eps)(s-i)$,
so we can run Fredman's algorithm on the remaining part of the permutation to partition it into
$(1+\epsilon)(s-i)$ decreasing subsequences.
Overall, this is at most $i+(1+\eps)(s-i) \leq (1+\eps)s$ subsequences.

On the lower bound side, we have the following.

\begin{lemma} \label{lem:DecompLower}
Any comparison-based algorithm for decomposing a permutation on $n$ elements into $\Oh(\sqrt{n})$
monotone subsequences needs $\Omega(n\log n)$ comparisons.
\end{lemma}

\begin{proof}
Assume that the algorithm creates at most $c\sqrt{n}$ subsequences using $\alpha n\log n$ comparisons.
Then, we can merge all subsequences into a single sorting sequence using $n\log(c\sqrt{n})$ comparisons
by merging them into pairs, quadruples, and so on. Thus, by the lower bound of $n\log n - \Oh(n)$ for any
comparison-based sorting algorithm, we must have $\alpha n\log n +n\log(c\sqrt{n}) \geq n\log n-\Oh(n)$,
so indeed $\alpha \geq 1/2-o(1)$.
\end{proof}

\section{Fully Dynamic Worst-Case Structure}
\label{sec:dynamic}

In this section we extend our decremental structure to also allow insertions of new elements
at any position in the main array, and then describe how to deamortise the fully dynamic structure.
As the required changes turn out to be relatively minor, we only provide a description of the new ingredients
and their effect on both the approximation guarantee and the running time.

\paragraph{Coordinates of points.}
Recall that our decremental structure was based on normalising the entries in the input
sequence $(a_{1},a_{2},\ldots,a_{n})$, creating a set of points $S=\{(i,a_{i}) : i\in [n]\}$, and
then maintaining some information for each nonempty dyadic rectangle.
Now that we allow both deletions and insertions, it is convenient to think that the points
have real coordinates, and we simply add new points to the current set without
modifying the coordinates of the already present points (provided that both the $x$-
and $y$-coordinates are distinct). However, we need to actually implement this efficiently.

We maintain a (totally ordered) list under deletions and insertions in an order-maintenance
structure. Such a structure can handle the following operations:
\begin{itemize}
\item $\mathsf{Insert}(a,b)$, which inserts element $a$ immediately after element $b$ in a list and returns a handle to $a$.
\item $\mathsf{Delete}(a)$.
\item $\mathsf{Order}(a,b)$, answers whether $a$ precedes $b$ in a list given handles of two elements.
\end{itemize}

\begin{theorem} \label{th:Order}
\cite{Order1,Order2} There is an order-maintenance data structure supporting all operations
in $\Oh(1)$ worst-case time.
\end{theorem}

To allow insertions in our algorithm, we maintain a global order-maintenance structure,
and $x$-coordinate of a point is a handle to an element in that structure.
Thus, we still are able to compare points in constant time.
Additionally, we keep in the structure two artificial guards corresponding to $-\infty$ and $\infty$ values on $x$-axis.
Our algorithm will sometimes still need an $x$-coordinate of an already deleted point,
therefore an element of the order-maintenance structure is deleted only when
no longer referred to. Nevertheless, it will be easy to verify that the order-maintenance structure stores
$\tilde{\Oh}(n)$ elements at any time.
We also note that using a BST to implement an order-maintenance structure in $\Oh(\log n)$ time
per operation would not dramatically increase our time complexities.

\paragraph{Two-dimensional recursion.}
While in the decremental structure we could use a fixed collection of rectangles,
doing so in the fully dynamic case is problematic, as now the number of points in a rectangle might either decrease or increase,
and we still need to be able to split a rectangle into either left/right or bottom/top
while maintaining logarithmic depth of the decomposition.
To this end, we define the rectangles differently.

The new definition is inspired by dynamic 2D range trees, see e.g.~\cite{Range}.
We first describe how to obtain the initial collection of rectangles, and then explain how to maintain it.
We apply a primary recursion on the $y$-coordinates, and then a secondary recursion on the $x$
coordinates, starting from a rectangle containing the whole set of points.
Take a rectangle $R$ containing $m$ points considered in the primary recursion.
We initially choose and store the middle $y$-coordinate, denoted $y_{m}$, such
that exactly $m/2$ points from $P(R)$ are below the horizontal line $y=y_m$,
then partition $R$ into $R_b$ and $R_t$ consisting of points under or above $y_m$, respectively,
thus each containing $m/2$ points. Then we repeat the primary recursion on $R_{b}$ and $R_{t}$.
Each rectangle considered in the primary recursion is then further partitioned
by the secondary recursion.
Take a rectangle $R$ containing $m$ points considered in the secondary recursion.
We initially choose and store the middle $x$-coordinate, denoted $x_{m}$, such
that exactly $m/2$ points from $P(R)$ are to the left of the vertical line $x=x_{m}$,
then partition $R$ into $R_{l}$ and $R_{r}$ consisting of points on left or right of $x=x_{m}$, respectively,
thus each containing $m/2$ points. Then we repeat the secondary recursion on $R_{l}$ and $R_{r}$.

To maintain the collection, we proceed as follows. Consider a rectangle $R$
created in the primary recursion with $m$ points, and let $y_{m}$ be its middle
$y$-coordinate. We keep $y=y_m$ as a line of partition of $R$ into $R_{b}$ and $R_{t}$ during
the next $m/4$ insertions and deletions of points \emph{inside} $R$, and then
recompute from scratch the whole recursive decomposition of $R$ (including running
the secondary recursion for each of the rectangles obtained in the primary recursion,
and constructing a covering family for each recomputed rectangle).
This is enough to guarantee that the partition of any $R$ into $R_{b}$ and $R_{t}$
in the primary recursion is $2$-balanced, meaning that the number of points
in $R_{b}$ and $R_{t}$ differs by at most a factor of $2$.
We proceed similarly for any rectangle considered in the secondary recursion,
so for any rectangle $R$ created in the secondary recursion for $m$ points, with $x_{m}$
being the middle $x$-coordinate, we keep $x=x_{m}$ as a line of partition of $R$ into $R_{l}$
and $R_{r}$ during the next $m/4$ insertions and deletions of points inside $R$, and then
recompute from scratch the whole recursive decomposition of $R$
(now this only refers to running the secondary recursion, and constructing
a covering family for each recomputed rectangle). The depth of both
recursions is clearly logarithmic, and the amortised time of recomputing
the decompositions is hopefully not too big.

The height of a node in a tree is the length of the longest path from that node to any leaf below it.
Each rectangle $R$ maintained by the algorithm is identified by a node of some secondary recursion tree $T$.
Additionally, we identify the root of $T$ with its corresponding node of the primary recursion tree,
so a rectangle $R$ identified with the root of $T$ is also identified with the corresponding node of the primary recursion tree.
Define the secondary height of $R$ to be the height of $R$ in its secondary recursion tree $T$,
and the primary height of $R$ to be the height of the root of $T$ in the primary recursion tree.
Since we keep $2$-balanced partition, both heights are at most $\log_{3/2} n = \Oh(\log n)$.
Because of this we need to slightly adjust $\eps'$ (but only up to a multiplicative constant).

In total, any point is inside $\Oh(\log^2 n)$ rectangles,
similarly to inside how many dyadic rectangles contained a given point in the previous sections.
Additionally, for a rectangle $R$ in the primary recursion tree with $|P(R)|=m$,
the sum of numbers of points inside rectangles below $R$ in the two-dimensional recursion is $\Oh(m \log^2 m)$.
For $R$ in the secondary recursion tree, this is only $\Oh(m \log m)$.
Recomputing in a rectangle is easy to amortise, but we need to explain how to
maintain $\CF(R)$ between recomputing of the whole rectangles, that is, how $\mathsf{Merge}$ works in this setting.

\begin{figure}[h]
\begin{center}
  \includegraphics[scale=1]{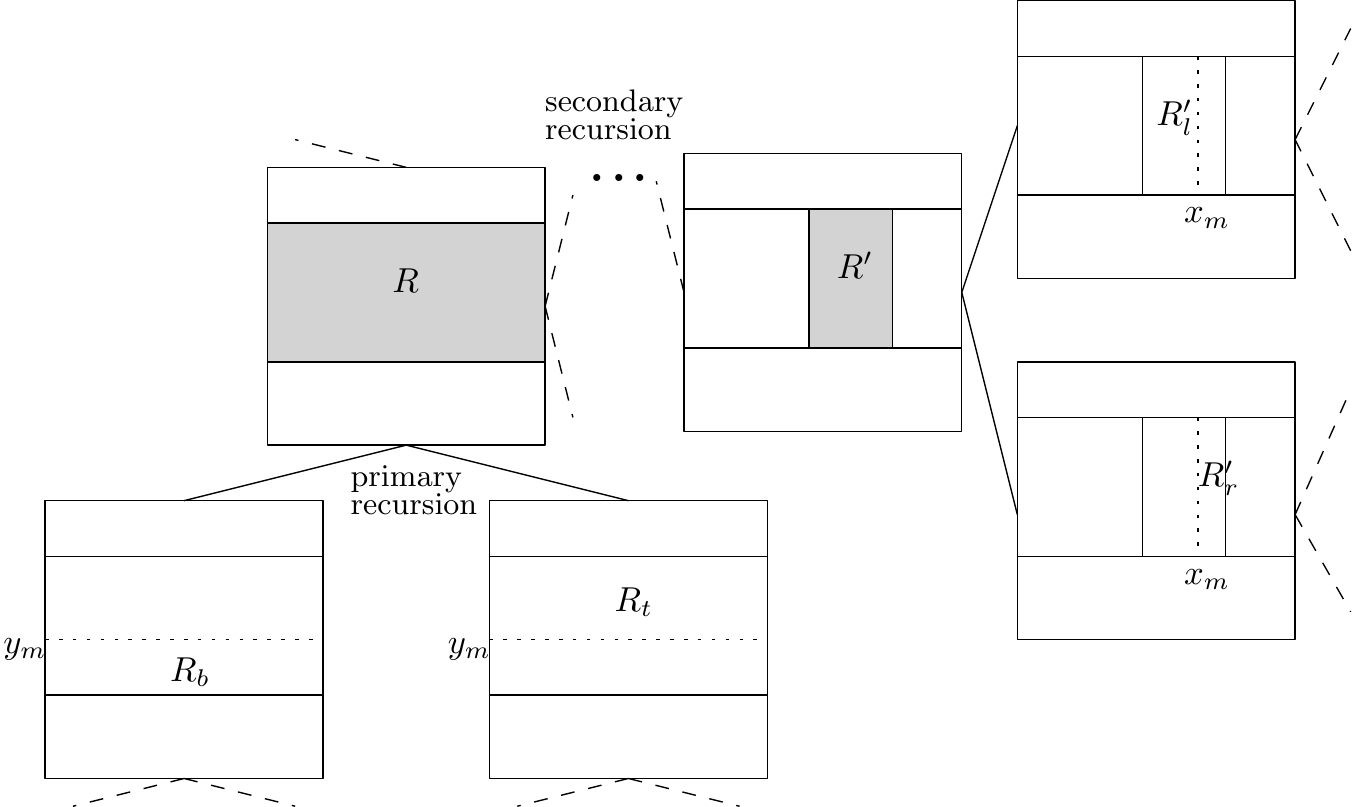}
\end{center}
  \caption{Rectangle $R$ in the primary recursion is partitioned into the bottom and top part.
  Additionally, it is identified with the root of a secondary recursion tree, where rectangles are repeatedly
  partitioned into the left and right part.
  }
  \label{fig:insert}
\end{figure}

\paragraph{$\mathsf{Merge}$ revisited.}
Let $R$ be some rectangle obtained in the primary recursion, in which
$R$ is divided into $2$-balanced $R_b$ and $R_t$ along the horizontal line $y=y_{m}$.
Let $R'$ be some rectangle obtained in the secondary recursion as a result of subdividing $R$.
$R'$ consists of all points in $P(R)$ with $x$-coordinates belonging to some interval,
and is divided into $2$-balanced $R'_l$ and $R'_r$ along the vertical line $x=x_{m}$.
Consult Figure~\ref{fig:insert}.
Say that a $(k_1,k_2)$-cover on level $l$ of $\CF(R')$ needs to be recomputed.
This can be achieved similarly as in the decremental structure.
We join BSTs of covers on level $l$ of $\CF(R'_l)$ and $\CF(R'_r)$,
and then run $\mathsf{Merge}$ using covering families of $R_b$ and $R_t$.
There are two minor issues.

The first issue is that $R_b$ and $R_t$ include all points from $P(R')$, but possibly contains many more,
as they form a partition of the whole $R$.
But that is not a problem, we can just discard all returned segments containing
points with $x$-coordinates outside of the interval of $R'$,
because if there is a chain of length $k_2$ inside $R'$, then it must be covered
by a chain of length $k_1$ also inside $R'$.
The analysis of correctness of $\mathsf{Merge}$ still holds.
The second issue is that while previously it was enough to analyse
only the depth of a set of segments returned by $\mathsf{Merge}$,
now we need to analyse the depth of the whole cover.

\begin{lemma}\label{lem:freq2}
For any rectangle $R$, the depth of any cover in $\CF(R)$ is $\Oh(\eps^{-1}\log^2 n)$.
\end{lemma}
\begin{proof}
We use induction on the height in the secondary recursion tree.
Let $h$ be the secondary height of $R$.
It was already shown in the proof of Lemma~\ref{lem:Tfreq} that, due to sparsification,
$\mathsf{Merge}$ adds only $\Oh(\eps^{-1}\log n)$ new segments, or at most
 $c' \eps^{-1}\log n$ for some constant $c'$.
Now we want to show that the depth of any cover in $\CF(R)$ is at most $1+c'h \eps^{-1}\log n$.
The base case of leaf rectangles containing just one point is trivial.
Otherwise, a cover in $\CF(R)$ is obtained by joining BSTs of the left and right rectangle,
then adding segments returned by $\mathsf{Merge}$ (and possibly deleting some segments along the way).
As the secondary height of both the left and right rectangle is at most $h-1$, and $\mathsf{Merge}$
returns no more than $c' \eps^{-1}\log n$ new segments, the depth of the obtained cover
is at most $1+c'h \eps^{-1}\log n$ as claimed.
Finally, as $h=\Oh(\log n)$ we obtain the lemma.
\end{proof}

\paragraph{Approximation factor.}
Here we analyse the approximation guarantee of the modified algorithm. 
Recall that a $(k_1,k_2)$-cover in the covering family of a rectangle $R$ is recomputed after roughly every
$\eps' k_1$ operations inside $R$.
In the decremental structure, the cover was calculated with a suitable margin,
namely by constructing a $(k_1\lam, k_2)$-cover, so that it remained a valid
$(k_1,k_2)$-cover even after $\eps' k_1$ deletions.
In the fully dynamic structure, we need to add such a margin on both sides.
Namely, we construct a $(k_1\lam, k_2/\lam)$-cover, which remains
a valid $(k_1, k_2)$-cover after any sequence of up to $\eps' k_1$ deletions or insertions.
This increases the approximation guarantee only by a constant factor, which
is formalised in the lemma below.
We defer its proof to the appendix, since it very much replicates that of Lemma~\ref{lem:approx}.

\begin{restatable}{lemma}{proofinserts} \label{lem:approx2}
For any rectangle $R'$ with primary height $h$, $\CF(R')$ is a $(\lam^2,3h)$-covering family.
\end{restatable}

Because of the above and also changed depth of the recursion, which is $\log_{3/2}{n}$ now,
we set $\eps'$ to be not $\eps/(8\log n)$ as before, but rather use
$\eps'=\eps/(12\log_{3/2} n)$.
This changes only the constants and we still have $1/\eps'=\Oh(\eps^{-1}\log{n})$.

\paragraph{Time complexity.}
Finally, we analyse the time complexity of the modified algorithm.
Observe that increasing the approximation factor and the height of the recursive structure
of rectangles only affects constants in the running time.
However, now $\mathsf{Merge}$ is slower by a factor of $\Oh(\log n)$,
as more segments cross the middle $x$-coordinate of a rectangle.
We also need to account for recomputing the covering family of a rectangle $R$
and every rectangle below in the recursive structure whenever the number of insertions and deletions
in $P(R)$ is sufficiently large.
This does not dominate the running time, though, as the cost of preprocessing is smaller than
the cost of maintaining covering families.
The time complexities are summarised in the lemma below.
We defer its proof to the appendix, since it very much replicates that of Lemma~\ref{lem:Ttotal}.

\begin{restatable}{lemma}{proofinsertstwo} \label{lem:Ttotal2}
The amortised complexity of deleting or inserting a point is $\Oh(\eps^{-5}\log^{11} n)$.
\end{restatable}

\paragraph{Deamortisation. } The final step is deamortising the fully dynamic structure.
Recall that we have used amortisation in three places: running $\mathsf{Merge}$ in order to maintain a cover,
recomputing the whole secondary recursive structure below rectangle $R'$,
and recomputing the whole primary recursive structure below $R$.
Additionally, we have implicitly assumed that the value of $\log n$ does not change
during the execution of the algorithm, so we need to rebuild the whole structure
once it changes by a constant factor (but this is standard, and we will not mention this issue again).
In all three places, we run the computation only once sufficiently many updates
have been performed in the corresponding structure (the structure being either a rectangle
in the primary/secondary recursion or a level of some covering family).
As usually, instead of performing the whole computation immediately,
we would like to distribute it among the next updates.

In the following, we assume we have a fixed value of $\eps$ and $n$ is the number of items in the main array;
as noted above, the value of $\log{n}$ does not change during the execution of the algorithm.
First, let us focus on some $(k_1,k_2)$-cover in a covering family of a rectangle $R$
and all calls to $\mathsf{Merge}$ used to maintain it.
In our solution, $\mathsf{Merge}$ works in $c_{n,\eps}(k_1)=\Oh(k_1\eps^{-3}\log^6{n})$ time,
and is performed after every $b_{n,\eps}(k_1)=\eps' k_1$ updates in $R$ (here we ignore additive constants).
After each update in $R$, we join trees from the left and right rectangle with the
segments returned by the most recent call to $\mathsf{Merge}$ in time $u_{n,\eps}(k_1)=\Oh(\eps^{-2} \log^4{n})$,
but note that we do not modify the set of segments returned by $\mathsf{Merge}$.
At any point, queries to the cover take time $q_{n,\eps}(k_1)=\Oh(\log{n})$.
This way, we achieved the amortised update cost of $u_{n,\eps}(k_1)+c_{n,\eps}(k_1)/b_{n,\eps}(k_1)=\Oh(\eps^{-4}\log^{7}{n})$.

Before we can deamortise this, we need to slightly modify the implementation of $\mathsf{Merge}$
due to the following reason. The computation needs to access some other
covering families and extract segments from their covers. However, because
now we want to run it in the background while updates to other structures are possibly taking
place, it is not clear what are the guarantees on the retrieved information.
Therefore, we always immediately execute the first part of $\mathsf{Merge}$ (up to line~\ref{candidates}) and gather
the relevant segments, each containing a list of elements that should be added to a set of points $P$.
This takes only $\Oh(\eps^{-2}\log^5 n)$ time, which is negligible.
We assume that these lists are not destroyed if pointed to in some part of the
structure, so we can think that the second part of $\mathsf{Merge}$ is a local computation that does
not need to access any other information. Similarly, recomputing the whole secondary
recursive structure below $R'$ and the whole primary recursive structure below $R$
can be implemented as a local computation, as each rectangle stores $P(R)$ in a persistent
BST (and no other information is required).

Now, to maintain a $(k_1,k_2)$-cover of $R$ in worst-case update time,
we use the current cover $A$ while constructing the next cover $B$ in the background.
The queries are answered using $A$, and the transition between $A$ and $B$ consists of two phases.
For the first $b_{n,\eps}(k_1)/3$ updates, nothing is done for $B$,
and for $A$ standard updates are performed in time $u_{n,\eps}(k_1)$.
Next, as mentioned above, in time $\Oh(\eps^{-2}\log^5 n)$ we immediately execute the first part of
$\mathsf{Merge}$, which gathers the relevant segments.
In the second phase, spanning the next $b_{n,\eps}(k_1)/3$ updates, we execute the remaining part of $\mathsf{Merge}$.
Namely, after every update in $A$, we run $3c_{n,\eps}(k_1)/b_{n,\eps}(k_1)$ steps of $\mathsf{Merge}$.
At the end of the second phase $B$ is ready, so it replaces $A$, which is no longer needed.
In total, $b_{n,\eps}(k_1)/3$ updates in $R$ passed from when we started constructing $B$
to when we started querying it.
Therefore, we still can query $B$ between the next $2b_{n,\eps}(k_1)/3$ updates needed to transition it into yet another cover.
Now, every update takes worst-case time $u_{n,\eps}(k_1)+\Oh(c_{n,\eps}(k_1)/b_{n,\eps}(k_1))=\Oh(\eps^{-4}\log^{7}{n})$, which is higher than amortised time only by a constant factor.

Next, we move to the local computation performed in the primary or secondary recursive structure,
which can be deamortised using a similar approach.
Let us focus on a rectangle in the primary recursive structure, as the other case is simpler,
with the only difference being smaller construction time.
Recall that we have a fixed value of $\eps$ and $n$ is the number of items in the main array.
Consider a rectangle initially storing $m \leq n$ items, for which the primary recursive structure
can be constructed in time $c_{n,\eps}(m)=\Oh(m\eps^{-2}\log^6{n})$,
and updated in time $u_{n,\eps}(m)=\Oh(\eps^{-5}\log^{9}{n})$
as long as the number of updates does not exceed $b_{n,\eps}(m)=m/4$.
Covers in the covering family of that rectangle can be queried in $q_{n,\eps}(m)=\Oh(\log{m})=\Oh(\log{n})$ time.
We obtained an amortised structure with no upper bound on the number of updates
by simply rebuilding it every $b_{n,\eps}(m)$ updates, resulting in
amortised update time $\Oh(c_{n,\eps}(2m)/b_{n,\eps}(m/2)+u_{n,\eps}(2m))$ and $q_{n,\eps}(m)$ query time,
where $m$ is the current number of items.

To obtain worst-case update time,
we maintain two structures $A$ and $B$, answering queries with $A$.
Assume that $A$ was initially constructed with $m_A$ points.
Whenever an update arrives, it is immediately executed in $A$.
Transition between $A$ and $B$ consists of three phases:
\begin{itemize}
\item For the first $b_{n,\eps}(m_A)/4$ updates, $B$ is not initialised yet.
Assume that after that, $A$ stores $m_B$ items.
\item In the second phase, we construct $B$ storing a set of those $m_B$ items, over at most $b_{n,\eps}(m_A)/8$ updates.
Namely, after every update in $A$, we run $8c_{n,\eps}(m_B)/b_{n,\eps}(m_A)$ steps of the construction algorithm
for $B$, if it has not finished work yet.
As $m_B < 2m_A$, this is $\Oh(c_{n,\eps}(2m_A)/b_{n,\eps}(m_A))$ steps after every update.
Additionally, in this and the next phase, every arriving update is added to
a queue of pending updates of $B$.
This is different from the case of $\mathsf{Merge}$, as now updates needs to be performed in both $A$ and $B$.
\item In the third phase, after every update in $A$, we execute two pending updates in $B$,
so after $b_{n,\eps}(m_A)/8$ updates the queue of pending updates becomes empty.
Then, $B$ replaces $A$.
\end{itemize}
In total, $b_{n,\eps}(m_A)/2$ updates are performed during the transition from $A$ to $B$.
Observe that in $B$, $b_{n,\eps}(m_A)/4$ updates were already performed,
so we can use $B$ for less than $b_{n,\eps}(m_B)$ another updates.
We defined $b_{n,\eps}(m)=m/4$, and due to the length of the first phase we have $m_B \geq m_A-b_{n,\eps}(m_A)/4 = 15m_A/16$.
Thus, it holds that:
\begin{eqnarray*}
b_{n,\eps}(m_B)-b_{n,\eps}(m_A)/4 &=& b_{n,\eps}(m_B)- m_A/16 \\
&\geq& m_B/4-m_B/15=11m_B/60 > m_B/8 = b_{n,\eps}(m_B)/2.
\end{eqnarray*}
This means we still can query $B$ during the next $b_{n,\eps}(m_B)/2$ updates needed to transition it into yet another structure.
Updates take worst-case time $\Oh(u_{n,\eps}(2m)+c_{n,\eps}(2m)/b_{n,\eps}(m/2))=\Oh(\eps^{-5}\log^{9}{n})$
and query is still in time $q_{n,\eps}(m)=\Oh(\log{n})$.

\section{Conclusions and Open Problems}
We constructed a dynamic algorithm providing arbitrarily small constant factor approximation
of the longest increasing subsequence, capable of deleting and inserting elements anywhere inside
the sequence in polylogarithmic time.
It improves on the previous result~\cite{Grids} providing these operations in time $\Oh(n^{\eps})$ 
and with a higher approximation factor.
We believe that our notion of a covering family, and the greedy procedure for obtaining one,
might be of independent interest.

Still, many aspects of this problem remain unexplored.
There are no known lower bounds on \emph{exact} dynamic algorithm for LIS better than $\Omega(\log n)$, 
while the best upper bound on exact dynamic LIS is $\Oh(n^{4/5})$~\cite{Kociumaka20}.
For approximated solutions, it should be relatively easy to decrease exponents in time complexity of our algorithm by one or two,
but it seems that achieving more practical time complexity, say $\Oh(\eps^{-3}\log^3 n)$,
would require a new approach or ideas.

\bibliographystyle{alpha}
\bibliography{LIS}

\appendix
\newpage

\section{Greedy algorithm for approximate cover}

We provide a full pseudocode of the greedy algorithm for computing $(k_1,k_2)$-cover.
Note that line~\ref{plusone} could be $i \gets q+1$, but this is not necessary.

\begin{algorithm}[h]
\begin{algorithmic}[1]
  \Function{Cover-approx}{$A,k_1,k_2$}
  \State \textbf{Input:} array $A = a_0,\ldots,a_{n-1}$, and two integer parameters, with $k_2>k_1$.
  \State \textbf{Output:} a collection of segments forming $(k_1,k_2)$-cover of $A$.
  \Statex
  \State $C \gets \emptyset$
  \State $i \gets 0$
  \While{$i < n$}
    \State $j \gets i$
    \While{\textcolor{blue}{$|\mathrm{LIS}(i,j)| < k_2$} and $j<n$}
      \State $j \gets j+1$
    \EndWhile
    \If{$j=n$}
    \State \Return $C$
    \EndIf
    \State $\triangleright$ $(i,j)$ is the shortest prefix with a chain of length $k_2$
    \State $q \gets j$
    \While{\textcolor{blue}{$|\mathrm{LIS}(q,j)| < k_1$}}
      \State $q \gets q-1$
    \EndWhile
    \State$\triangleright$ $(q,j)$ is the shortest suffix with a chain of length $k_1$
    \State $C \gets C \cup \{(q,j,\mathrm{LIS}(q,j)\}$,
    \State \textcolor{blue}{$i \gets q$} \label{plusone}
  \EndWhile
  \State \Return $C$
  \EndFunction
\end{algorithmic}
\caption{The greedy algorithm computing one level of an approximate cover.}
\label{alg:greedy_approx}
\end{algorithm}

\newpage
\section{Deferred Proofs}

\proofinserts*
\begin{proof}
This is proved by induction first on the primary height, then the secondary height.
Rectangles of zero height contain a single point and the claim is trivial.
Let us now consider rectangle $R'$ of primary height $h$, with left/right rectangles $R'_l$/$R'_r$.
We also have bottom/top rectangles $R_b$/$R_t$ constituting a partition
of the rectangle $R$ corresponding to the root of the secondary recursion tree of $R'$.
From the assumption, $R'_l$ and $R'_r$ provide $(\lam^2,3h)$-covering families,
while $R_b$ and $R_t$ provide $(\lam^2,3h-3)$-covering families.
The first $k=3/\eps'$ levels in $\CF(R')$ are exact covers and they are correct by Lemma~\ref{lem:merge2}.

Let us focus on a level providing a  $(\lceil k\lam^{2j} \rceil, k\lam^{2(j+3h-1)})$-cover for some $j>0$.
Segments coming from the same level in $R'_l$ and $R'_r$ have the correct score.
We need to consider segments coming from $\mathsf{Merge}(\CF(R_b),\CF(R_t),x_{m},k\lam^{2j+2},\normalfont{\texttt{true}}$),
where $x_{m}$ is the middle $x$-coordinate of $R'$.
$\CF(R_b)$ and $\CF(R_t)$ provide $\lam^{6h-6}$-approximation.
Say that we have an interval crossing $x_{m}$ and containing a chain $X$ of length at least $k\lam^{2j+6h-3}$.
By Lemma~\ref{lem:merge1}, in the set returned by $\mathsf{Merge}$ we have a chain of length
$\lceil k\lam^{2j+1} \rceil$ covering $X$.

Because of amortising $\mathsf{Merge}$, there might be at most
$\lfloor \eps' \lceil k\lam^{2j} \rceil \rfloor -2$ deletions or insertions of points in $R'$ before
the whole level of the covering family is recomputed.
Thus, the guarantee on the length of a chain is at least $\lceil k\lam^{2j}\rceil$, as in Lemma~\ref{lem:approx}.

If after $\lfloor \eps' \lceil k\lam^{2j} \rceil \rfloor -2$ deletions or insertions of points in $R'$
there exists an interval crossing $x_{m}$ and containing chain of length $k\lam^{2(j+3h-1)}$,
then at the time of recomputing this interval contained a chain of length at least $k\lam^{2j+6h-3}$, as
$k\lam^{2j+6h-2}-\lfloor \eps' \lceil k\lam^{2j} \rceil \rfloor +2 >
\lam^{6h-2}k\lam^{2j}-(\lam-1)k\lam^{2j} > k\lam^{2j}(\lam^{6h-2}-\lam+1) > k\lam^{2j+6h-3}$
for $1<\lam<2$.
\end{proof}

\proofinsertstwo*
\begin{proof}
When deleting an element, the algorithm needs to visit $\Oh(\log^2 n)$ rectangles that contain it.
For each such rectangle $R$, we do the following:
\begin{itemize}
\item Delete the element from the BST storing points inside $R$ in time $\Oh(\log n)$.
\item Increment the counter of every level of $\CF(R)$, there are $\Oh(\eps^{-1}\log^2 n)$ levels.
\item Pay the amortised cost of recomputing every level of $\CF(R)$, now in total this is
$\Oh(\eps^{-5}\log^9 n)$ instead of $\Oh(\eps^{-5}\log^8 n)$, because the depth of the middle $x$-coordinate is
$\Oh(\eps^{-1}\log^2 n)$ instead of $\Oh(\eps^{-1}\log n)$.
\item For every level of $\CF(R)$, we join BSTs from $\CF(R_l),\CF(R_r)$ with the segments returned by 
the most recent call to $\mathsf{Merge}$.
This takes time $\Oh(\eps^{-2}\log^4 n)$ in total.
\item Check if there have been sufficiently many insertions or deletions of points in $R$, and if so
recompute the whole recursive structure below $R$.
If $R$ contains $m$ points, then as in Lemma~\ref{lem:Tprep} this is done in time $\Oh(m\eps^{-2}\log^6 n)$
(even less if $R$ is not a rectangle in the primary recursion tree).
The recomputation is done after $\Theta(m)$ deletions or insertions, thus the amortised time is $\Oh(\eps^{-2}\log^6 n)$.\qedhere
\end{itemize}
\end{proof}

\end{document}